\documentclass[twoside,11pt]{article}

\usepackage{jmlr2e-style/jmlr2e}

\usepackage{hhline}
\usepackage{appendix}
\usepackage{amsmath}
\usepackage{amsfonts}
\usepackage{amssymb}
\usepackage{array}
\usepackage{multirow}
\usepackage{caption}
\usepackage{subcaption}

\usepackage{enumitem}
\setlist{nosep}

\usepackage{booktabs}
\usepackage{makecell}
\usepackage{pdflscape}
\usepackage{xcolor}

\newcommand{\superhard}{\textit{Super~Hard}}
\newcommand{\ultrahard}{\textit{Ultra~Hard}}

\newcommand{\statespace}{\mathbb{S}}
\newcommand{\state}{s}
\newcommand{\timespace}{\mathbb{T}}
\newcommand{\horizon}{T}
\newcommand{\timestep}{t}
\newcommand{\jointactionspace}{\mathbb{U}}
\newcommand{\jointaction}{\mathbf{u}}
\newcommand{\action}{u}
\newcommand{\statetransitionfunction}{P}
\newcommand{\rewardfunction}{R}
\newcommand{\reward}{r}
\newcommand{\jointobservationspace}{\mathbb{O}}
\newcommand{\jointobservation}{\mathbf{o}}
\newcommand{\observation}{o}
\newcommand{\jointobservationhistoryspace}{\mathbb{H}}
\newcommand{\jointobservationhistory}{\mathbf{h}}
\newcommand{\observationhistory}{h}
\newcommand{\observationfunction}{O}
\newcommand{\agentspace}{\mathbb{K}}
\newcommand{\agentcounter}{k}
\newcommand{\numberofagents}{K}
\newcommand{\discountfactor}{\gamma}
\newcommand{\policy}{\pi}
\newcommand{\tderror}{\delta}

\newcommand{\huberloss}{\mathcal{L}_\kappa}
\newcommand{\asymmetrichuberloss}{\rho}
\newcommand{\loss}{\mathcal{L}}
\newcommand{\utilityexp}{Q}
\newcommand{\utilityadv}{A}
\newcommand{\utilitystate}{V}
\newcommand{\utility}{Z}
\newcommand{\joint}{\mathrm{jt}}
\newcommand{\quantilefunction}{F^{-1}}

\newcommand{\quantile}{\omega}
\newcommand{\pdf}{f}
\newcommand{\cdf}{F}
\newcommand{\pdfmodelparameter}{\alpha}
\newcommand{\modelparameter}{\beta}
\newcommand{\inversecdf}{F^{-1}}
\newcommand{\numberofquantiles}{N}
\newcommand{\numberofquantilesamples}{N'}
\newcommand{\numberofquantiletestsamples}{\hat{N}}
\newcommand{\implicitquantilefunction}{g}
\newcommand{\stateembeddingfunction}{\psi}
\newcommand{\cosineembeddingfunction}{\phi}
\newcommand{\cosinedimension}{m}
\newcommand{\monotonicfunction}{M}
\newcommand{\meandecompositionfunction}{\Psi}
\newcommand{\shapedecompositionfunction}{\Phi}

\newcommand{\additivity}{\textit{Additivity}}
\newcommand{\monotonicity}{\textit{Monotonicity}}
\newcommand{\diql}{DIQL}
\newcommand{\ddn}{DDN}
\newcommand{\dmix}{DMIX}
\newcommand{\dplex}{DPLEX}
\newcommand{\digm}{DIGM}
\newcommand{\eqd}{\overset{D}{=}}
\newcommand{\iter}{i}

\definecolor{myblue}{HTML}{0000B5}
\definecolor{crimson}{HTML}{B30000}
\definecolor{indigo}{HTML}{4B0082}
\newcommand{\bb}[1]{\textcolor{myblue}{#1}}
\newcommand{\cc}[1]{\textcolor{crimson}{#1}}

\newcounter{example0}
\newcounter{theorem0}
\newcounter{lemma0}
\newcounter{proposition0}
\newcounter{remark0}
\newcounter{corollary0}
\newcounter{definition0}
\newcounter{conjecture0}
\newcounter{axiom0}

\newtheorem{example}[example0]{Example} 

\newtheorem{theorem}[theorem0]{Theorem}

\newtheorem{lemma}[lemma0]{Lemma} 

\newtheorem{proposition}[proposition0]{Proposition} 

\newtheorem{remark}[remark0]{Remark}

\newtheorem{corollary}[corollary0]{Corollary}

\newtheorem{definition}[definition0]{Definition}

\newtheorem{conjecture}[conjecture0]{Conjecture}

\newtheorem{axiom}[axiom0]{Axiom}

\usepackage{lastpage}
\jmlrheading{24}{2023}{1-\pageref{LastPage}}{6/22; Revised
2/23}{5/23}{22-0630}{Wei-Fang Sun, Cheng-Kuang Lee, Simon See, and Chun-Yi Lee}

\ShortHeadings{A Unified Framework for Factorizing Distributional Value Functions}{Sun, Lee, See, and Lee}
\firstpageno{1}

\begin{document}

\title{A Unified Framework for Factorizing Distributional Value Functions for Multi-Agent Reinforcement Learning}

\author{\name Wei-Fang Sun\thanks{Wei-Fang Sun contributed to the work during his NVIDIA internship.} \email j3soon@gapp.nthu.edu.tw \\
       \addr Elsa Lab, Department of Computer Science\\
       National Tsing Hua University\\
       No. 101, Section 2, Kuang-Fu Road, Hsinchu City 30013, Taiwan
       \AND
       \name Cheng-Kuang Lee \email ckl@nvidia.com \\
       \addr NVIDIA AI Technology Center\\
       NVIDIA Corporation\\
       Santa Clara, CA 94305 USA
       \AND
       \name Simon See \email ssee@nvidia.com \\
       \addr NVIDIA AI Technology Center\\
       NVIDIA Corporation\\
       Santa Clara, CA 94305 USA
       \AND
       \name Chun-Yi Lee\thanks{Corresponding Author.} \email cylee@cs.nthu.edu.tw \\
       \addr Elsa Lab, Department of Computer Science\\
       National Tsing Hua University\\
       No. 101, Section 2, Kuang-Fu Road, Hsinchu City 30013, Taiwan
}

\editor{Shipra Agrawal}

\maketitle

\begin{abstract}
In fully cooperative multi-agent reinforcement learning (MARL) settings, environments are highly stochastic due to the partial observability of each agent and the continuously changing policies of other agents. 
To address the above issues, we proposed a unified framework, called DFAC, for integrating distributional RL with value function factorization methods. This framework generalizes expected value function factorization methods to enable the factorization of return distributions. To validate DFAC, we first demonstrate its ability to factorize the value functions of a simple matrix game with stochastic rewards. Then, we perform experiments on all \superhard{} maps of the StarCraft Multi-Agent Challenge and six self-designed \ultrahard{} maps, showing that DFAC is able to outperform a number of baselines.
\end{abstract}


\section{Introduction}
\label{sec:introduction}

In multi-agent reinforcement learning (MARL), one of the popular research directions is to enhance the training procedure of fully cooperative and decentralized agents. Examples of such agents include a fleet of unmanned aerial vehicles (UAVs), a group of autonomous cars, etc. This research direction aims to develop a decentralized and cooperative behavior policy for each agent, and is especially difficult for MARL settings without an explicit communication channel. The most straightforward approach is independent Q-learning (IQL)~\citep{Tan1993IQL}, where each agent is trained independently with its own behavior policy aimed to optimize the total return. Nevertheless, each agent's policy may not converge owing to two main difficulties: (1) non-stationary environments caused by the changing behaviors of other agents, and (2) spurious reward signals originated from the actions of other agents. In addition, an agent’s partial observability of an environment further exacerbates the above issues.
On the other hand, another line of research focuses on fully utilizing all agents' information to learn a joint cooperative policy with a single monolithic network. Such methods, named fully-centralized learners, eliminate the policy convergence issue. Unfortunately, they are not scalable and are unable to be executed in a decentralized fashion. In order to strike a balance between these two types of approaches, a number of MARL researchers turned their attention to centralized training with decentralized execution (CTDE) methods, with an objective to stabilize the training procedure while maintaining the agents' abilities for decentralized execution~\citep{Oliehoek2016CTDE}.
Among these CTDE approaches, value function factorization methods~\citep{Sunehag2018VDN,Rashid2018QMIX,Son2019QTRAN} are especially promising in terms of their superior performances and data efficiency~\citep{Samvelyan2019SMAC}.


Value function factorization methods aim to learn a factorization of a certain joint value function based on all agents' information during centralized training, so as to enable decentralized execution based on the learned factorization.
These methods introduce the assumption of individual-global-max (IGM)~\citep{Son2019QTRAN}, which assumes that each agent's optimal actions result in the optimal joint actions of the entire group.
Based on IGM, the total return of a group of agents can be factorized into separate utility functions~\citep{Guestrin2001Utility} (or simply `\textit{utility}' hereafter) for each agent. The utilities allow the agents to independently derive their own optimal actions during execution, and deliver promising performance in StarCraft Multi-Agent Challenge (SMAC)~\citep{Samvelyan2019SMAC}. Unfortunately, current value function factorization methods only concentrate on estimating the expectations of the utilities, overlooking the additional information contained in the full return distributions. Such information, nevertheless, has been demonstrated beneficial for policy learning in the recent literature~\citep{Lyle2019Comparative}.


In the past few years, distributional RL has been empirically shown to enhance value function estimation in various single-agent RL (SARL) domains~\citep{Bellemare2017C51,Dabney2018QR-DQN,Dabney2018IQN,Rowland2019ER-DQN,Yang2019FQF}. Instead of estimating a single scalar Q-value, it approximates the probability distribution of a total return by either a categorical distribution~\citep{Bellemare2017C51} or a quantile function~\citep{Dabney2018QR-DQN,Dabney2018IQN}. Even though the above methods may be beneficial to the MARL domain due to the ability to capture uncertainty, it is inherently incompatible to expected value function factorization methods (e.g., value decomposition network (VDN)~\citep{Sunehag2018VDN}, monotonic mixing network (QMIX)~\citep{Rashid2018QMIX}, and duplex dueling mixing network (QPLEX)~\citep{Wang2020QPLEX}). The incompatibility arises from two aspects: (1) maintaining IGM in a distributional form, and (2) factorizing the probability distribution of a total return into individual utilities. As a result, an effective and efficient approach that is able to solve the incompatibility is crucial and necessary for bridging the gap between value function factorization methods and distributional RL.


In this paper, we propose a \textbf{D}istributional Value Function \textbf{Fac}torization (DFAC) framework, to efficiently integrate value function factorization methods with distributional RL.
DFAC solves the incompatibility by two techniques: (1) Mean-Shape Decomposition and (2) Quantile Mixture. The former allows the generalization of expected value function factorization methods (e.g., VDN, QMIX, and QPLEX) to their DFAC variants without violating IGM. The latter allows a total return distribution to be factorized into individual utility distributions in a computationally efficient manner.  We slightly abuse the term `factorizing return distributions' here, which refer to factorizing the value function into individual utilities, and should not be confused with factorizing a joint distribution into a product of marginal distributions.
To validate the effectiveness of DFAC, we first demonstrate the ability of factorizing the return distribution of a matrix game with stochastic rewards.
Then, we perform experiments on all \superhard{} maps in SMAC, and show that DFAC offers beneficial impacts on the baseline methods.
Furthermore, we designed six \ultrahard{} maps to further validate the effectiveness of DFAC. In summary, the primary contribution is the introduction of DFAC for bridging the gap between distributional RL and value function factorization methods efficiently by mean-shape decomposition and quantile mixture. Please note that this work is an extended version of~\cite{sun21dfac}.

\section{Related Works}
\label{sec:sup:related_work}

The related works of this paper can be grouped into four categories as follows.

\textbf{Distributional Independent Learners.} Independent learners allow agents to learn a cooperative policy based on their own action-observation histories, without assuming additional information from other agents.
The major drawback of this approach is the miscoordination among agents due to the lack of information exchange during training. Several works incorporate distrbutional RL methods to mitigate this issue.
\cite{Da2019IQLC51} extended the simplest IQL~\citep{Tan1993IQL}, which simply performs expected Q-learning for each agent, to multiple independent C51 learners~\citep{Bellemare2017C51}. This approach performs distributional Q-learning for each agent, and is able to deliver superior results.
\cite{Lyu2020LikelihoodQuantile} combined Hysteretic Q-Learning (HQL)~\citep{Matignon2007Hysteretic,omidshafiei2017deep} with implicit quantile networks (IQN)~\citep{Dabney2018IQN} to separate the impacts of miscoordination among agents and the environmental stochasticity. \cite{Rowland2021TemporalDA} proposed a framework to investigate the difference between asymmetric update methods such as HQL, and distributional RL methods such as  expectile distributional reinforcement learning (EDRL)~\citep{Rowland2019ER-DQN}.

\textbf{CTDE Value Function Factorization.} Value function factorization methods aim to learn a factorization of the joint value function for a group of agents during centralized training. The learned factorization allows the policy of each agent to be executed in a decentralized fashion.
The first proposed method in this category is VDN~\citep{Sunehag2018VDN}. VDN assumes the \additivity{} property, which constrains a joint action-value function to be the sum of multiple individual utilities. As a result, VDN employs a joint value network to accomplish such an objective. 
QMIX~\citep{Rashid2018QMIX} extends VDN to a larger set of tasks by assuming the \monotonicity{} property, which requires a joint action-value function to be equivalent to at least one of the monotonic transformations of individual utilities. Based on this relaxed assumption, QMIX utilizes a monotonic mixing network to combine individual utilities in a monotonic fashion. QTRAN~\citep{Son2019QTRAN} further transforms the joint action-value function into another easily factorizable one by introducing two additional regularization losses, and allows an even larger set of tasks to be factorized correctly. However, QTRAN is unstable during training~\citep{Rashid2018QMIX}, which motivates several subsequent works dedicated to extending QMIX to solve the same set of tasks as QTRAN while maintaining the training stability. For example, W-QMIX~\citep{Rashid2020WQMiX} treats QMIX as an operator, and projects non-monotonic tasks to QMIX's problem domain. Qatten~\citep{Yang2020Qatten} introduces a multi-head attention mechanism to approximate a joint action-value function. QPLEX~\citep{Wang2020QPLEX} utilizes a duplex dueling mixing network for stabilizing QTRAN's training process.

\textbf{Other CTDE Methods.} Aside from the general CTDE value function factorization methods, a number of methods have been proposed to introduce additional assumptions to further enhance the general CTDE methods. They include the use of intrinsic reward signals~\citep{Du2019LIIR}, explicit communication channels~\citep{Zhang2019VBN,Wang2019NDQ}, common knowledge shared among agents~\citep{De2019MACKRL,Wang2020ROMA}, prior knowledge about environments~\citep{Wang2020ASN}, and so on. In addition, there has been a line of research that focuses on an actor-critic style training. For example, Multi-Agent Deep Deterministic Policy Gradient (MADDPG)~\citep{Lowe2017MADDPG} uses a centralized critic to improve the value function approximation ability. The critic is augmented with extra information about all agents' policies, while the actor only has access to its local information. Counterfactual Multi-Agent Policy Gradients (COMA)~\citep{Foerster2018COMA} leverages a centralized critic that approximates a baseline for each agent. The baseline enables an agent to marginalize out its action while keeping the other agents' actions fixed, allowing it to reason about counterfactual scenarios to facilitate its learning. Multi-Actor-Attention-Critic (MAAC)~\citep{Iqbal2019MAAC} constructs a critic for each agent with an attention mechanism, such that the critics are able to selectively pay attention to the information from the other agents. Moreover, Factored Multi-Agent Centralized Policy Gradients (FACMAC)~\citep{Peng2021FACMAC} incorporates value function factorization methods into the centralized critic in MADDPG to further improve its performance. Please note that in this paper, we focus on improving value-based estimation for CTDE training. These methods are orthogonal to our work.

\textbf{Distributional Value Function Factorization.} Similar to our work, QR-MIX~\citep{Hu2020QRMIX} and RMIX~\citep{Anonymous2020RMIX} also utilized distributional RL to increase their performance in SMAC. QR-MIX is a QMIX variant which models the mixing network as an IQN. Instead of decomposing the joint distribution into individual utility distributions, it treats the utilities as Dirac distributions. On the other hand, RMIX is also a QMIX variant, but replaces the inputs and outputs of the mixing network with the Conditional Value at Risk (CVaR) measures. The major difference between our work and the above two methods is the ability to factorize the joint value distribution into individual utility distributions, while these two methods are not capable of factorizing distributions.

\section{Background}
\label{sec:background}

In this section, we introduce the essential background material for understanding the contents of this paper. We first define the problem formulation of cooperative MARL and CTDE. Next, we describe the conventional formulation of IGM and the value function factorization methods. Then, we walk through the concepts of distributional RL, quantile function, as well as quantile regression, which are the fundamental concepts frequently mentioned in this paper. After that, we explain IQN, a key approach adopted in this paper for approximating quantiles.  Finally, we bring out the concept of quantile mixture, which is leveraged by DFAC for factorizing return distributions.

\subsection{Cooperative MARL and CTDE}
\label{subsec:background_cooperative_marl_and_ctde}

In this work, we consider a fully cooperative MARL environment modeled as a decentralized and partially observable Markov Decision Process (Dec-POMDP)~\citep{Oliehoek2016CTDE} with stocastic rewards, which is described as a tuple $\langle\statespace{},\agentspace{},\jointobservationspace_{\joint},\jointactionspace_{\joint},\statetransitionfunction{},\observationfunction{},\rewardfunction{},\discountfactor{}\rangle$ and is defined as follows:

\begin{itemize}
\setlength\itemsep{0em}

\item $\statespace{}$ is the finite set of global states in the environment, where $\state{}'\in \statespace{}$ denotes the next state of the current state $\state{}\in \statespace{}$. The state information is optionally available during training, but is not available to the agents during execution.

\item $\agentspace{}=\{1,...,\numberofagents{}\}$ is the set of $\numberofagents{}$ agents. We use $\agentcounter{}\in\agentspace{}$ to denote the index of the agent.

\item $\jointobservationspace_{\joint}=\Pi_{\agentcounter{}\in\agentspace{}}\jointobservationspace{}_{\agentcounter{}}$ is the set of joint observations. At each timestep, a joint observation $\jointobservation{}=\langle\observation{}_1,...,\observation{}_\numberofagents{}\rangle\in\jointobservationspace_{\joint}$ is provided by the environment to the $\numberofagents{}$ agents. Each agent $\agentcounter{}$ is only able to observe its own individual observation $\observation{}_{\agentcounter{}}\in\jointobservationspace{}_{\agentcounter{}}$.

\item $\jointobservationhistoryspace_{\joint}=\Pi_{\agentcounter{}\in\agentspace{}}\jointobservationhistoryspace{}_{\agentcounter{}}$ is the set of joint action-observation histories. The joint history $\jointobservationhistory{}=\langle\observationhistory{}_1,...,\observationhistory{}_\numberofagents{}\rangle\in\jointobservationhistoryspace_{\joint}$ concatenates all the perceived observations and the performed actions before a certain timestep, where $\observationhistory{}_{\agentcounter{}}\in\jointobservationhistoryspace{}_{\agentcounter{}}$ represents the action-observation history from agent $\agentcounter{}$.

\item $\jointactionspace_{\joint}=\Pi_{\agentcounter{}\in\agentspace{}}\jointactionspace{}_{\agentcounter{}}$ is the set of joint actions. At each timestep, the entire group of the agents take a joint action $\jointaction{}$, where $\jointaction{}=\langle\action{}_1,...,\action{}_{\numberofagents{}}\rangle\in\jointactionspace_{\joint}$. The individual action $\action{}_{\agentcounter{}}\in\jointactionspace_{\agentcounter{}}$ of each agent $\agentcounter{}$ is determined based on its stochastic policy $\policy{}_{\agentcounter{}}(\action{}_{\agentcounter{}}|\observationhistory{}_{\agentcounter{}}): \jointobservationhistoryspace{}_{\agentcounter{}}\times \jointactionspace{}_{\agentcounter{}}\rightarrow[0,1]$, expressed as $\action{}_{\agentcounter{}}\sim\policy{}_{\agentcounter{}}(\cdot|\observationhistory{}_{\agentcounter{}})$. Similarly, in single agent scenarios, we use $\action$ and $\action'$ to denote the actions of the agent at state $\state$ and $\state'$ under its policy $\policy$, respectively.

\item $\timespace{}=\{1,...,\horizon{}\}$ represents the set of timesteps with horizon $\horizon{}$, where the index of the current timestep is denoted as $\timestep{}\in\timespace{}$. Please note that $\state^{\timestep}$, $\jointobservation^{\timestep}$, $\jointobservationhistory^{\timestep}$, and $\jointaction^{\timestep}$ correspond to the state, joint observation, joint action-observation history, and joint action
at timestep $\timestep$, respectively.

\item The transition function $\statetransitionfunction{}(\state{}'|\state{},\jointaction{}):\statespace{}\times\jointactionspace_{\joint}\times \statespace{}\rightarrow[0,1]$ specifies the state transition probabilities. Given $\state$ and $\jointaction$, the next state is represented as $\state{}'\sim\statetransitionfunction{}(\cdot|\state{},\jointaction{})$.

\item The observation function $\observationfunction{}({\jointobservation{}}|\state{}):\jointobservationspace_{\joint}\times\statespace{}\rightarrow[0,1]$ specifies the joint observation probabilities. Given $\state{}$, the joint observation is represented as $\jointobservation{}\sim\observationfunction{}(\cdot|\state{})$.

\item $\rewardfunction{}(\reward|\state{},\jointaction{}):\statespace{}\times\jointactionspace_{\joint}\times\mathbb{R}\rightarrow[0,1]$ is the joint reward function shared among all the agents. Given $\state$, the team reward is expressed as $\reward\sim\rewardfunction{}(\cdot|\state{},\jointaction{})$.

\item $\discountfactor{}\in\mathbb{R}$ is the discount factor with its value within $[0, 1)$.
\end{itemize}
Under such an MARL formulation, this work concentrates on CTDE-based value function factorization methods, where the agents are trained in a centralized fashion and executed in a decentralized manner. In other words, the joint action-observation history $\jointobservationhistory{}$ is available during the learning processes of individual policies $[\policy{}_{\agentcounter{}}]_{\agentcounter{}\in\agentspace{}}$. During execution, each agent's policy $\policy{}_{\agentcounter{}}$ only conditions on its observation history $\observationhistory{}_{\agentcounter{}}$.

\subsection{IGM and Factorizable Tasks}
\label{subsec:background_igm_and_factorizable_task}

IGM is necessary for value function factorization~\citep{Son2019QTRAN}. For a joint action-value function $\utilityexp{}_{\joint{}}(\jointobservationhistory{},\jointaction{}): \jointobservationhistoryspace_{\joint} \times \jointactionspace_{\joint} \rightarrow\mathbb{R}$, if there exist $\numberofagents{}$ separate utility functions $[\utilityexp{}_\agentcounter{}(\observationhistory{}_\agentcounter{},\action{}_\agentcounter{}): \jointobservationhistoryspace{}_{\agentcounter{}} \times \jointactionspace{}_{\agentcounter{}} \rightarrow\mathbb{R}]_{\agentcounter{}\in\agentspace{}}$ such that the following condition holds:
\begin{equation}
\arg\max_\jointaction{} \utilityexp{}_{\joint{}}(\jointobservationhistory{},\jointaction{}) =
\begin{pmatrix}
\arg\max_{\action{}_1} \utilityexp{}_1(\observationhistory{}_1,\action{}_1)\\
\vdots \\
\arg\max_{\action{}_\numberofagents{}} \utilityexp{}_\numberofagents{}(\observationhistory{}_\numberofagents{},\action{}_\numberofagents{})
\end{pmatrix},
\label{eq:igm}
\end{equation}
then $[\utilityexp{}_\agentcounter{}]_{\agentcounter{}\in\agentspace{}}$ are said to satisfy IGM for $\utilityexp{}_{\joint{}}$ under $\jointobservationhistory{}$. 
Under this condition, we also say that $\utilityexp{}_{\joint{}}(\jointobservationhistory{},\jointaction{})$ is factorized by $[\utilityexp{}_\agentcounter{}(\observationhistory{}_\agentcounter{},\action{}_\agentcounter{})]_{\agentcounter{}\in\agentspace{}}$~\citep{Son2019QTRAN}. If $\utilityexp{}_{\joint{}}$ in a given task is factorizable under all $\jointobservationhistory{}\in \jointobservationhistoryspace_{\joint}$, we say that the task is factorizable. Intuitively, factorizable tasks indicate that there exists a factorization such that each agent can select the greedy action according to their individual utilities $[\utilityexp{}_\agentcounter{}]_{\agentcounter{}\in\agentspace{}}$ independently in a decentralized fashion. This enables the optimal individual actions to implicitly achieve the optimal joint action across the $\numberofagents{}$ agents. Since there is no individual reward, the factorized utilities do not estimate expected returns on their own \citep{Guestrin2001Utility} and are different from the value function definition commonly used in SARL.


\subsection{Advantage-based IGM}
\label{subsec:advantage_based_igm}

In addition to Q-value based IGM described in Section~\ref{subsec:background_igm_and_factorizable_task}, there exist a variant, called \textit{advantage-based IGM}~\citep{Wang2020QPLEX}, which is defined as the following:
\begin{equation}
\arg\max_\jointaction{} \utilityadv{}_{\joint{}}(\jointobservationhistory{},\jointaction{}) =
\begin{pmatrix}
\arg\max_{\action{}_1} \utilityadv{}_1(\observationhistory{}_1,\action{}_1)\\
\vdots \\
\arg\max_{\action{}_\numberofagents{}} \utilityadv{}_\numberofagents{}(\observationhistory{}_\numberofagents{},\action{}_\numberofagents{})
\end{pmatrix},
\label{eq:advantage-igm}
\end{equation}
where the advantage functions $[A_\agentcounter{}(\observationhistory{}_\agentcounter{},\action{}_\agentcounter{})]_{\agentcounter{}\in\agentspace{}}$ can be derived from $[Q_\agentcounter{}(\observationhistory{}_\agentcounter{},\action{}_\agentcounter{})]_{\agentcounter{}\in\agentspace{}}$ and the state-value functions $[V_\agentcounter{}(\observationhistory{}_\agentcounter{})]_{\agentcounter{}\in\agentspace{}}$ as:

\begin{equation}
A_\agentcounter{}(\observationhistory{}_\agentcounter{},\action{}_\agentcounter{})=Q_\agentcounter{}(\observationhistory{}_\agentcounter{},\action{}_\agentcounter{})-V_\agentcounter{}(\observationhistory{}_\agentcounter{}), \forall\agentcounter{}\in\agentspace{},
\end{equation}
\begin{equation}
V_\agentcounter{}(\observationhistory{}_\agentcounter{})=\max_{\action{}_\agentcounter{}}Q_\agentcounter{}(\observationhistory{}_\agentcounter{},\action{}_\agentcounter{}), \forall\agentcounter{}\in\agentspace{}.
\end{equation}
Eq.~(\ref{eq:advantage-igm}) has been proved to be equivalent to Eq.~(\ref{eq:igm})~\citep{Wang2020QPLEX}, since the argmax operations over the actions do not depend on $[V_\agentcounter{}(\observationhistory{}_\agentcounter{})]_{\agentcounter{}\in\agentspace{}}$, but only on $[A_\agentcounter{}(\observationhistory{}_\agentcounter{},\action{}_\agentcounter{})]_{\agentcounter{}\in\agentspace{}}$.


\subsection{Value Function Factorization Methods}
\label{subsec:background_value_factorization_methods}

Based on IGM, value function factorization methods enable centralized training for factorizable tasks, while maintaining the ability for decentralized execution. In this work, we consider three such methods: VDN, QMIX, and QPLEX, where their expressive power (or expressiveness) can be sorted as: QPLEX $>$ QMIX $>$ VDN. A factorization function being more expressive means that it is able to solve a larger subset of factorizable tasks.
VDN solves a subset of factorizable tasks that satisfies \additivity{} by defining the factorization function as the sum of individual utilities:
\begin{equation} \utilityexp{}_{\joint{}}(\jointobservationhistory{},\jointaction{}) = \sum^{\numberofagents{}}_{\agentcounter{}=1} \utilityexp{}_\agentcounter{}(\observationhistory{}_\agentcounter{},\action{}_\agentcounter{}).
\label{eq:additivity}
\end{equation}
QMIX solves a larger subset of factorizable tasks that satisfies \monotonicity{} by defining the factorization function as a monotonic combination of the individual utilities:
\begin{equation} \utilityexp{}_{\joint{}}(\jointobservationhistory{},\jointaction{}) = \monotonicfunction(\utilityexp{}_1(\observationhistory{}_1,\action{}_1),...,\utilityexp{}_\numberofagents{}(\observationhistory{}_\numberofagents{},\action{}_\numberofagents{})\vert\state),
\label{eq:monotonicity}
\end{equation}
where $\monotonicfunction$ is a learnable monotonic mixing network that satisfies $\frac{\partial \monotonicfunction}{\partial \utilityexp{}_\agentcounter{}}\ge 0, \forall \agentcounter{}\in\agentspace{}$, and can condition on the state $\state$ if the information is available during training.
QPLEX solves all factorizable tasks~\citep{Wang2020QPLEX} by defining the factorization function as follows:
\begin{equation}
\utilityexp{}_{\joint{}}(\jointobservationhistory{},\jointaction{}) = \sum^{\numberofagents{}}_{\agentcounter{}=1} \utilitystate_\agentcounter{}(\jointobservationhistory{}) + \sum^{\numberofagents{}}_{\agentcounter{}=1}\lambda_\agentcounter{}(\jointobservationhistory{},\jointaction{}) \utilityadv_\agentcounter{}(\jointobservationhistory{},\action{}_\agentcounter{}).
\label{eq:qplex}
\end{equation}
where $\utilitystate_\agentcounter{}(\jointobservationhistory{})=\monotonicfunction_\agentcounter{}(\utilitystate_\agentcounter{}(\observationhistory{}_\agentcounter{})\vert\state)$, $\utilityadv_\agentcounter{}(\jointobservationhistory{},\action{}_\agentcounter{})=\monotonicfunction_\agentcounter{}(\utilityadv_\agentcounter{}(\observationhistory{}_\agentcounter{},\action{}_\agentcounter{})\vert\state)$, and $\lambda_\agentcounter{}(\jointobservationhistory{},\jointaction{})>0$. In addition to the learnable monotonic networks $[\monotonicfunction_\agentcounter]_{\agentcounter{}\in\agentspace{}}$, such a factorization further learns a positive weight $\lambda_\agentcounter{}$ for each individual advantage. This flexibility enables QPLEX to correct the discrepancy between the joint Q-value and the individual utilities when a task does not satisfy \additivity{}.
All of these three factorization functions satisfy the IGM condition~\citep{Son2019QTRAN}, which indicates that the individual utilities $[\utilityexp{}^{(t)}_\agentcounter{}]_{\agentcounter{}\in\agentspace{}}$ satisfy IGM for $\utilityexp{}^{(t)}_{\joint{}}$ under all $\jointobservationhistory{}$ throughout the entire training process (i.e., for all $t$). The joint Q-value $\utilityexp{}_{\joint{}}$ is learned iteratively by updating the joint network with the Bellman optimality operator $\mathcal{T}^*$ based on the joint reward signal:
\begin{align}
    \utilityexp{}^{(\iter+1)}_{\joint{}}(\jointobservationhistory{},\jointaction{}) \leftarrow \mathcal{T}^*\utilityexp{}^{(\iter)}_{\joint{}}(\jointobservationhistory{},\jointaction{}),
\label{eq:joint_q_update}
\end{align}
where $\iter$ denotes the index of iterations, $\mathcal{T}^*\utilityexp{}^{(\iter)}_{\joint{}}(\jointobservationhistory{},\jointaction{}) = \mathbb{E}[\rewardfunction{}(\state{},\jointaction{})]+\gamma \utilityexp{}^{(\iter)}_{\joint{}}(\state{}',\jointaction{}'^*)$, and $\jointaction{}'^*=\arg\max_{\jointaction'}\utilityexp{}^{(\iter)}_{\joint{}}(\state{}',\jointaction')$ is the optimal action at state $\state{}'$. Eq.~(\ref{eq:joint_q_update}) enables the joint network to be trained in the same fashion as the single-agent DQN~\citep{Mnih2015DQN}, allowing $\utilityexp{}_{\joint{}}$ to approximate the optimal joint Q-value function $\utilityexp{}^*_{\joint{}}$ while satisfying the IGM condition. This update process of minimizing the Temporal Difference (TD) error implicitly realizes an effective counterfactual credit assignment by factorizing $\utilityexp{}_{\joint{}}$ into individual utilities $[\utilityexp{}_\agentcounter{}]_{\agentcounter{}\in\agentspace{}}$~\citep{wang2021towards}.
The factorized utilities are then utilized to determine each agent's policy during execution.

The representational capacity of the joint network is determined by two factors: (1) the network's (or the model's) capacity, which is limited by the number of learnable parameters, and (2) the expressiveness of a factorization function, which is dependent on the selected factorization method.
For example, QPLEX has been proved to be more expressive than QMIX~\citep{Wang2020QPLEX}. Due to the theoretical limitation of QMIX's expressiveness, there exist certain tasks that QMIX are unable to solve by simply minimizing the TD errors of the joint Q-values, even if QMIX is equipped with unlimited model capacity.


\subsection{Distributional RL}
\label{subsec:background_distributional_rl}

For notational simplicity, we consider a degenerated case with only a single agent, and the environment is fully observable until the end of Section~\ref{subsec:background_implicit_quantile_network}. Distributional RL generalizes classic expected RL methods by capturing the full return distribution $\utility{}(\state{},\action{})$ instead of the expected return $\utilityexp{}(\state{},\action{})$, allowing an agent to learn better state representations~\citep{BDR2022} that lead to better performance in various single-agent domains~\citep{Bellemare2017C51,Bellemare2019S51,Dabney2018QR-DQN,Dabney2018IQN,Rowland2019ER-DQN,Yang2019FQF,Nguyen2021MMDRL}. Modeling the full return distributions can be viewed as an auxiliary task that benefits representation learning during training, where the information other than the expected returns can be discarded afterward. In addition, distributional RL enables applications that require the information of the full return distributions, such as exploration~\citep{Nikolov2019IDS,Zhang2019QUOTA,Mavrin2019DLTV} and risk-aware policies~\citep{Xia2020risk}. Distributional RL defines the distributional Bellman operator $\mathcal{T}^\pi$ as follows:
\begin{equation}
\mathcal{T}^\pi \utility{}(\state{},\action{})\overset{D}{=}\rewardfunction{}(\state{},\action{})+\gamma \utility{}(\state{}',\action{}'),
\end{equation}
and the distributional Bellman optimality operator $\mathcal{T}^*$ as:
\begin{equation}
\mathcal{T}^* \utility{}(\state{},\action{})\overset{D}{=}\rewardfunction{}(\state{},\action{})+\gamma \utility{}(\state{}',\action{}'^*),
\end{equation}
where $\action{}'^*=\arg\max_{\action'}\mathbb{E}[\utility{}(\state{}',\action')]$ is the optimal action at state $\state{}'$, and the expression $X\overset{D}{=}Y$ denotes that random variables $X$ and $Y$ follow the same distribution. Given some initial distribution $Z_0$, applying different Bellman operators lead to different results.
If $\mathcal{T}^\pi$ is applied repeatedly, $\utility{}$ converges to $\utility{}^\pi$ in $p$-Wasserstein distance for all $p\in[1,\infty)$ under $\pi$.
On the other hand, if $\mathcal{T}^*$ is applied instead, $\utility{}$ alternates among the optimal return distributions in the set $\mathcal{\utility{}}^*:=\{\utility^{\pi^*}:\pi^*\in\Pi^*\}$, where $\Pi^*$ denotes the set of optimal policies~\citep{Bellemare2017C51}. The $p$-Wasserstein distance between the probability distributions of random variables $X$, $Y$ is defined as:
\begin{equation}
W_p(X,Y)=\left(\int_0^1|\inversecdf_X(\quantile)-\inversecdf_Y(\quantile)|^p \mathrm{d}\quantile\right)^{1/p},
\end{equation}
where $(\inversecdf_X,\inversecdf_Y)$ are quantile functions of $(X,Y)$.


\subsection{Categorical Distribution and Heuristic Projection}

The distributional Bellman optimality operator $\mathcal{T}^*$ lays the foundation for C51~\citep{Bellemare2017C51}, the first distributional RL algorithm. It models the return distribution as a categorical distribution, where the number of atoms $n$ (i.e., categories) within the range $[V_\mathrm{MIN}, V_\mathrm{MAX}]$ and the distance between them $\triangle z=\frac{V_\mathrm{MAX}-V_\mathrm{MIN}}{n-1}$ are hyperparameters.
Since C51 can only approximate the return distributions with a finite support (i.e., the set of atoms: $\{z_i = V_\mathrm{MIN} + i\triangle z : 0\le i < n\}$), a heuristic projection is required to project the Bellman target distribution onto the pre-defined support.
Specifically, C51 aims to minimize the KL divergence:
\begin{equation}
\loss{}(\state{},\action{},\reward{},\state{}') = D_{\mathrm{KL}}(\mathrm{proj}(\mathcal{T}^* \utility{}(\state{},\action{}))||\utility{}(\state{},\action)),
\end{equation}
where the heuristic projection function $\mathrm{proj}(\cdot)$ approximates $\mathcal{T}^* \utility{}(\state{},\action{})$ by distributing the probability of each atoms to their immediate neighbors on the support of C51:
\begin{equation}
[\mathrm{proj}(\mathcal{T}^* \utility{}(\state{},\action{}))]_i = \sum_{j=0}^{n-1} \left[1-\frac{|[\mathcal{T}^* z_j]^{V_\mathrm{MAX}}_{V_\mathrm{MIN}} - z_i|}{\triangle z}\right]^1_0 p_j(\state{}',\action{}'^*),
\label{eq:heuristic_projection}
\end{equation}
where $[\cdot]_a^b$ is a clipping function that bounds its argument in $[a,b]$, and $p_j(\state{}',\action{}'^*)$ is the probability mass of the $j^{\mathrm{th}}$ atom. Unfortunately, the heuristic projection used in Eq.~(\ref{eq:heuristic_projection}) introduces variance~\citep{Rowland2019ER-DQN} and does not guarantee convergence~\citep{Dabney2018QR-DQN}. 
In addition, C51 usually requires a large number of categories, and often necessitates re-adjusting the support range for tasks with different reward magnitudes, making it inappropriate for practical usages.
Therefore, its follow-up works turn to approximate return distributions with quantile functions instead of categorical distributions.


\subsection{Quantile Function and Quantile Regression}
\label{subsec:background_quantile_function_and_quantile_regression}

The relationship between the cumulative distribution function (CDF) $F_X$ and the quantile function $\quantilefunction_X$ (i.e., the generalized inverse CDF) of a random variable $X$ is formulated as:
\begin{equation}
\quantilefunction_X(\quantile{})=\inf\{x\in\mathbb{R}:\quantile{}\le \cdf{}_X(x)\}, \forall\quantile\in[0,1],
\label{eq:quantile_function_to_inverse_cdf}
\end{equation}
where $\quantile$ represents the quantile. The expectation of $X$ expressed in terms of $\quantilefunction_X(\quantile{})$ is:
\begin{equation}
\mathbb{E}[X]=\int_0^1\quantilefunction_X(\quantile{})\ \mathrm{d}\quantile.
\label{eq:expectation_of_quantile_function}
\end{equation}
\cite{Dabney2018IQN} model the value function as a quantile function $\quantilefunction{}(\state{},\action{}\vert\quantile{})$,
and use a pair-wise sampled temporal difference (TD) error $\tderror$ for two quantile samples $\quantile{}, \quantile{}'\sim U([0,1])$ to optimize the value function. The TD error $\tderror$
is defined as:
\begin{equation}
\label{eq:quantile_TD_error}
\tderror^{\quantile{},\quantile{}'}=\reward+\gamma \quantilefunction{}(\state{}',\action{}'^*\vert\quantile{}') - \quantilefunction{}(\state{},\action{}\vert\quantile{}).
\end{equation}
Based on Eq.~(\ref{eq:quantile_TD_error}), the pair-wise loss $\asymmetrichuberloss{}^\kappa_{\quantile{}}$ is formulated based on the Huber quantile regression loss $\huberloss{}$~\citep{Dabney2018QR-DQN} with threshold $\kappa=1$, and is expressed as follows:
\begin{equation}
\asymmetrichuberloss{}^\kappa_{\quantile{}}(\tderror^{\quantile{},\quantile{}'})=|\quantile{}-\mathbb{I}\{\tderror^{\quantile{},\quantile{}'}<0\}|\frac{\huberloss{}(\tderror^{\quantile{},\quantile{}'})}{\kappa} \text{, where}
\end{equation}
\begin{equation}
\huberloss{}(\tderror^{\quantile{},\quantile{}'})=
\begin{cases}
    \frac{1}{2}(\tderror^{\quantile{},\quantile{}'})^2, & \text{if }|\tderror^{\quantile{},\quantile{}'}|\le\kappa\\
    \kappa(|\tderror^{\quantile{},\quantile{}'}|-\frac{1}{2}\kappa), & \text{otherwise}
\end{cases}.
\end{equation}
Given $\numberofquantiles{}$ quantile samples $[\quantile{}_i]_{i=1}^{\numberofquantiles}$ to be optimized with regard to $\numberofquantilesamples$  target quantile samples $[\quantile{}_j]_{j=1}^{\numberofquantilesamples}$, the total loss $\loss{}(\state{},\action{},\reward{},\state{}')$  is defined as the sum of the pair-wise losses, and is expressed as the following:
\begin{equation}
\loss{}(\state{},\action{},\reward{},\state{}')=\frac{1}{\numberofquantilesamples{}}\sum_{i=1}^{\numberofquantiles{}}\sum_{j=1}^{\numberofquantilesamples{}}\asymmetrichuberloss{}^\kappa_{\quantile{}_i}(\tderror^{\quantile{}_i,\quantile{}_j'}).
\label{eq:total_qr_loss}
\end{equation}


\subsection{Implicit Quantile Network}
\label{subsec:background_implicit_quantile_network}

In order to realize the concepts discussed in the previous section, implicit quantile network (IQN)~\citep{Dabney2018IQN} proposes to approximate the return distribution $\utility{}(\state{},\action{})$ by an implicit quantile function $\quantilefunction{}(\state{},\action{}\vert\quantile{})=\implicitquantilefunction{}(\stateembeddingfunction{}(\state{}),\cosineembeddingfunction{}(\quantile{}))_\action{}$ for functions $\implicitquantilefunction{}$, $\psi$, and $\phi$. The subscript $\action{}$ denotes a certain action, which corresponds to the $\action{}^\mathrm{th}$ element in the output vector of $\implicitquantilefunction{}$. IQN is a type of universal value function approximator (UVFA)~\citep{Schaul2015UVFA}.  It employs a light-weight architecture which generalizes its predictions across states $\state{}\in\statespace{}$ and goals $\quantile{}\in [0,1]$, with the goals defined as different quantiles of the return distribution.
To mitigate \textit{spectral bias} \citep{rahaman2019spectral}, $\cosineembeddingfunction{}$ first maps the input scalar $\quantile$ to an $\cosinedimension$-dimensional vector by a high frequency function $[\cos(\pi i\quantile)]^{\cosinedimension{}-1}_{i=0}$, followed by a single hidden layer with weights $[w_{ij}]$ and biases $[b_j]$ to produce a quantile embedding $\cosineembeddingfunction{}(\quantile{})=[\cosineembeddingfunction(\quantile{})_j]^{\text{dim}(\cosineembeddingfunction{}(\quantile{}))-1}_{j=0}$. The expression of $\cosineembeddingfunction(\quantile{})_j$ can be represented as the following:
\begin{equation}
\cosineembeddingfunction(\quantile{})_j:=\text{ReLU}(\sum_{i=0}^{\cosinedimension{}-1}\cos(\pi i\quantile)w_{ij}+b_j),
\end{equation}
where $\cosinedimension{}=64$. Then, $\cosineembeddingfunction{}(\quantile{})$ is combined with the state embedding $\stateembeddingfunction{}(\state{})$ by the element-wise (Hadamard) product ($\odot$), expressed as $\implicitquantilefunction:=\stateembeddingfunction{}\odot\cosineembeddingfunction{}$, where $\text{dim}(\cosineembeddingfunction(\quantile{}))=\text{dim}(\stateembeddingfunction(\state{}))$. The loss of IQN is defined as Eq.~(\ref{eq:total_qr_loss}) by sampling a batch of $\numberofquantiles{}$ and $\numberofquantilesamples{}$ quantiles from the policy network and the target network respectively.
During execution, the action with the largest expected return $\utilityexp{}(\state{},\action{})$ is chosen. Since IQN does not model the expected return explicitly, $\utilityexp{}(\state{},\action{})$ is approximated by calculating the mean of the sampled returns through $\numberofquantiletestsamples{}$ quantile samples $\hat\quantile{}_i\sim U([0,1]), \forall i\in[1,\numberofquantiletestsamples{}]$ based on Eq.~(\ref{eq:expectation_of_quantile_function}), expressed as follows:
\begin{equation}
\utilityexp{}(\state{},\action{})=\int_0^1\quantilefunction{}(\state{},\action{}\vert\quantile{})\ \mathrm{d}\quantile\approx\frac{1}{\numberofquantiletestsamples{}}\sum_{i=1}^{\numberofquantiletestsamples{}}\quantilefunction{}(\state{},\action{}\vert\hat\quantile{}_i).
\end{equation}


\subsection{Quantile Mixture}
\label{subsec:background_quantile_mixture}

Multiple quantile functions (e.g., IQNs) sharing the same quantile $\quantile{}$ may be combined into a single quantile function $\quantilefunction{}(\quantile{})$, in a form of quantile mixture expressed as follows:
\begin{equation}
\quantilefunction{}(\quantile{})=\sum^{\numberofagents{}}_{k=1}\modelparameter{}_{\agentcounter{}} \quantilefunction_{\agentcounter{}}(\quantile{}),
\label{eq:quantile_mixture}
\end{equation}
where $[\quantilefunction_{\agentcounter{}}(\quantile{})]_{\agentcounter{}\in\agentspace{}}$ are quantile functions, and $[\modelparameter{}_{\agentcounter{}}]_{\agentcounter{}\in\agentspace{}}$ are model parameters~\citep{Karvanen2006QuantileMixture}. The condition for  $[\modelparameter{}_{\agentcounter{}}]_{\agentcounter{}\in\agentspace{}}$ is that $\quantilefunction{}(\quantile{})$ must satisfy the properties of a quantile function.
The concept of quantile mixture must not be confused with the mixture of multiple probability density functions (PDFs):
\begin{equation}
\pdf{}(x)=\sum^{\numberofagents{}}_{\agentcounter{}=1}\pdfmodelparameter{}_{\agentcounter{}} \pdf{}_{\agentcounter{}}(x),
\end{equation}
where $[\pdf{}_{\agentcounter{}}(x)]_{\agentcounter{}\in\agentspace{}}$ are PDFs, $\sum^{\numberofagents{}}_{\agentcounter{}=1}\pdfmodelparameter{}_{\agentcounter{}}=1$, and $\pdfmodelparameter{}_{\agentcounter{}}\ge 0$. Summing PDFs is equivalent to making a random choice from a collection of random variables, whereas adding quantile functions combines the random variables with a specific dependence structure (Eq.~(\ref{eq:quantile_mixture})).

\section{Methodology}
\label{sec:methodology}

In this section, we walk through the proposed DFAC framework and its derivation procedure. We first discuss a naive distributional factorization approach and its limitation in Section~\ref{subsec:methodology_distributional_igm_condition}.
Then, we introduce the DFAC framework to address the limitation, and show that DFAC is able to generalize distributional RL to all factorizable tasks in Section~\ref{subsec:methodology_the_proposed_dfac_framework}.
After that, several practical implementations of DFAC are presented in Section~\ref{subsec:methodology_a_practical_implementation_of_dfac}.
Finally, \ddn{}, \dmix{}, and \dplex{}, which are the DFAC variants of VDN, QMIX, and QPLEX, respectively, are discussed in Section~\ref{subsec:methodology_distributional_variant_of_vdn_and_qmix}.


\subsection{Distributional IGM}
\label{subsec:methodology_distributional_igm_condition}

Since the IGM condition is required for value function factorization to  work in fully cooperative tasks, a distributional factorization that satisfies IGM is essential for factorizing return distributions. We first discuss a naive distributional factorization approach that simply replaces deterministic utilities $\utilityexp$ with stochastic utilities $\utility$. Then, we provide a theorem to show that the naive approach is insufficient to guarantee the IGM condition for all factorization functions in general.
\begin{definition}[Distributional IGM]
\label{def:distributional_igm}

A finite number of individual stochastic utilities\\ $[\utility_{\agentcounter}(\observationhistory_{\agentcounter}, \action_{\agentcounter})]_{\agentcounter \in \agentspace}$ are said to satisfy Distributional IGM (\digm{}) for a stochastic joint action-value function $\utility_{\joint}(\jointobservationhistory, \jointaction{})$ under $\jointobservationhistory$, if their expectations $[\mathbb{E}[\utility_{\agentcounter}(\observationhistory_{\agentcounter}, u_{\agentcounter})]]_{\agentcounter \in \agentspace}$ satisfy IGM for $\mathbb{E}[\utility_{\joint}(\jointobservationhistory, \jointaction{})]$ under $\jointobservationhistory$, represented as follows:
\begin{equation}
\arg\max_\jointaction{} \mathbb{E}[\utility{}_{\joint{}}(\jointobservationhistory{},\jointaction{})] =
\begin{pmatrix}
\arg\max_{\action{}_1} \mathbb{E}[\utility{}_1(\observationhistory{}_1,\action{}_1)]\\
\vdots \\
\arg\max_{\action{}_\numberofagents{}} \mathbb{E}[\utility{}_\numberofagents{}(\observationhistory{}_\numberofagents{},\action{}_\numberofagents{})]
\end{pmatrix}.
\notag{}
\end{equation}
\end{definition}
By replacing $\utilityexp$ in Eq.~(\ref{eq:igm}) with $\mathbb{E}[\utility]$, IGM can be extended to Distributional IGM. The Advantage-based IGM can also be extended to its distributional version in a similar way. In the following, we use $\utilityexp$ and $\mathbb{E}[\utility]$ interchangeably, specifically, $\utilityexp_{\joint{}}(\jointobservationhistory{},\jointaction{}) = \mathbb{E}[\utility_{\joint{}}(\jointobservationhistory{},\jointaction{})]$, and $\utilityexp_{\agentcounter}(\observationhistory_{\agentcounter}, u_{\agentcounter}) = \mathbb{E}[\utility_{\agentcounter}(\observationhistory_{\agentcounter}, u_{\agentcounter})], \forall\agentcounter \in \agentspace$.
\begin{proposition}
\label{prop:digm_equals_igm}

\digm{} is equivalent to IGM (Eq.~(\ref{eq:igm})).
\end{proposition}

\begin{proof}
\begin{equation}
\begin{aligned}
\arg\max_\jointaction{} \mathbb{E}[\utility{}_{\joint{}}(\jointobservationhistory{},\jointaction{})] &=
\begin{pmatrix}
\arg\max_{\action{}_1} \mathbb{E}[\utility{}_1(\observationhistory{}_1,\action{}_1)]\\
\vdots \\
\arg\max_{\action{}_\numberofagents{}} \mathbb{E}[\utility{}_\numberofagents{}(\observationhistory{}_\numberofagents{},\action{}_\numberofagents{})]
\end{pmatrix}\\
\Leftrightarrow\arg\max_\jointaction{} \utilityexp{}_{\joint{}}(\jointobservationhistory{},\jointaction{}) &=
\begin{pmatrix}
\arg\max_{\action{}_1} \utilityexp{}_1(\observationhistory{}_1,\action{}_1)\\
\vdots \\
\arg\max_{\action{}_\numberofagents{}} \utilityexp{}_\numberofagents{}(\observationhistory{}_\numberofagents{},\action{}_\numberofagents{})
\end{pmatrix}.
\end{aligned}
\notag{}
\end{equation}
\end{proof}
The Advantage-based IGM is also equivalent to its distributional version, which can be proved in the same manner.
Following the above definition of DIGM,
we next discuss the reason why simply replaces $\utilityexp$ with $\utility$ (i.e., the naive approach) is insufficient to guarantee the DIGM condition by considering \additivity{} (in Eq.~(\ref{eq:additivity})) and \monotonicity{} (in Eq.~(\ref{eq:monotonicity})) in two propositions.
\begin{proposition}
\label{prop:distributional_additivity}

\additivity{} for utility distributions:
\begin{equation}
\utility{}_{\joint{}}(\jointobservationhistory{},{\jointaction{}}) \eqd \meandecompositionfunction{}(\utility{}_1(\observationhistory{}_1,\action{}_1), ..., \utility{}_\numberofagents(\observationhistory{}_\numberofagents,\action{}_\numberofagents)) \eqd \sum_{\agentcounter{}\in\agentspace{}} \utility{}_{\agentcounter{}}(\observationhistory{}_{\agentcounter{}},\action{}_{\agentcounter{}})
\notag{}
\end{equation}
is a sufficient condition for DIGM.

\end{proposition}

\begin{proof}
By linearity of expectation, the following equation holds:
\begin{equation}
\begin{aligned}
\mathbb{E}[\utility{}_{\joint{}}(\jointobservationhistory{},{\jointaction{}})]
&= \mathbb{E}[\sum_{\agentcounter{}\in\agentspace{}} \utility{}_{\agentcounter{}}(\observationhistory{}_{\agentcounter{}},\action{}_{\agentcounter{}})] = \sum_{\agentcounter{}\in\agentspace{}} \mathbb{E}[\utility{}_{\agentcounter{}}(\observationhistory{}_{\agentcounter{}},\action{}_{\agentcounter{}})] \\
\Leftrightarrow \utilityexp{}_{\joint{}}(\jointobservationhistory{},{\jointaction{}}) &= \sum_{\agentcounter{}\in\agentspace{}} \utilityexp{}_{\agentcounter{}}(\observationhistory{}_{\agentcounter{}},\action{}_{\agentcounter{}}).
\end{aligned}
\notag{}
\end{equation}
Since \additivity{} is a sufficient condition for IGM:
\begin{equation}
\begin{split}
\mathbb{E}[\utility{}_{\joint{}}(\jointobservationhistory{},{\jointaction{}})] &= \meandecompositionfunction{}(\mathbb{E}[\utility{}_1(\observationhistory{}_1,\action{}_1)], ..., \mathbb{E}[\utility{}_\numberofagents(\observationhistory{}_\numberofagents,\action{}_\numberofagents)]) = \sum_{\agentcounter{}\in\agentspace{}} \mathbb{E}[\utility{}_{\agentcounter{}}(\observationhistory{}_{\agentcounter{}},\action{}_{\agentcounter{}})] \\
\Leftrightarrow \utilityexp{}_{\joint{}}(\jointobservationhistory{},{\jointaction{}}) &= \meandecompositionfunction{}(\utilityexp{}_1(\observationhistory{}_1,\action{}_1), ..., \utilityexp{}_\numberofagents(\observationhistory{}_\numberofagents,\action{}_\numberofagents)) = \sum_{\agentcounter{}\in\agentspace{}} \utilityexp{}_{\agentcounter{}}(\observationhistory{}_{\agentcounter{}},\action{}_{\agentcounter{}}),
\end{split}
\notag{}
\end{equation}
$[\mathbb{E}[\utility_{\agentcounter}(\observationhistory_{\agentcounter}, u_{\agentcounter})]]_{\agentcounter \in \agentspace}$ satisfy IGM for $\mathbb{E}[\utility_{\joint}(\jointobservationhistory, \jointaction{})]$ under $\jointobservationhistory$. According to Definition~\ref{def:distributional_igm} and Proposition~\ref{prop:digm_equals_igm}, $[\utility_{\agentcounter}(\observationhistory_{\agentcounter}, u_{\agentcounter})]_{\agentcounter \in \agentspace}$ satisfy \digm{} for $\utility_{\joint}(\jointobservationhistory, \jointaction{})$ under $\jointobservationhistory$.
As a result, \additivity{} for utility distributions is a sufficient condition for \digm{}.
\end{proof}
Simply replacing $\utilityexp$ with $\utility$ (i.e., the naive approach) satisfies \digm{} for \additivity{} based on linearity of expectations. However, linearity of expectations does not hold in general for monotonic transformations:
\begin{equation}
\begin{split}
\utilityexp{}_{\joint{}}(\jointobservationhistory{},{\jointaction{}}) &= \mathbb{E}[\utility{}_{\joint{}}(\jointobservationhistory{},{\jointaction{}})] \\
&= \mathbb{E}[\monotonicfunction{}(\utility{}_1(\observationhistory{}_1,\action{}_1), ..., \utility{}_{\numberofagents{}}(\observationhistory{}_{\numberofagents{}},\action{}_{\numberofagents{}})\vert\state)], \forall\monotonicfunction{} \\
&\ne \monotonicfunction{}(\mathbb{E}[\utility{}_1(\observationhistory{}_1,\action{}_1)], ..., \mathbb{E}[\utility{}_{\numberofagents{}}(\observationhistory{}_{\numberofagents{}},\action{}_{\numberofagents{}})]\vert\state), \forall\monotonicfunction{} \\
&= \monotonicfunction{}(\utilityexp{}_1(\observationhistory{}_1,\action{}_1), ..., \utilityexp{}_{\numberofagents{}}(\observationhistory{}_{\numberofagents{}},\action{}_{\numberofagents{}})\vert\state), \forall\monotonicfunction{}.
\end{split}
\notag{}
\end{equation}
Next, we provide a counterexample to demonstrate this property.
\begin{proposition}
\label{prop:distributional_monotonicity}

\monotonicity{} for utility distributions:
\begin{equation}
\begin{split}
\utility{}_{\joint{}}(\jointobservationhistory{},{\jointaction{}}) &\eqd \meandecompositionfunction{}(\utility{}_1(\observationhistory{}_1,\action{}_1), ..., \utility{}_\numberofagents(\observationhistory{}_\numberofagents,\action{}_\numberofagents)\vert\state)\\
&\eqd \monotonicfunction{}( \utility{}_1(\observationhistory{}_1,\action{}_1), ..., \utility{}_{\numberofagents{}}(\observationhistory{}_{\numberofagents{}},\action{}_{\numberofagents{}})\vert\state),
\end{split}
\notag{}
\end{equation}
where $\monotonicfunction$ is a monotonic transformation that satisfies $\frac{\partial \monotonicfunction}{\partial \utilityexp{}_\agentcounter{}}\ge 0, \forall \agentcounter{}\in\agentspace{}$, is not a sufficient condition for DIGM, even though the condition may satisfy \digm{} for special cases of $\monotonicfunction{}$ and $[\utility{}_{\agentcounter{}}(\observationhistory{}_{\agentcounter{}},\action{}_{\agentcounter{}})]_{\agentcounter{}\in\agentspace{}}$.

\end{proposition}

\begin{proof}
We consider a degenerated case and prove the theorem by contradiction. Consider a case where there is only a single agent ($\numberofagents{}=1$), with a single fully observable state and an exponential transformation $\monotonicfunction{}(\utility{}_1(\observationhistory{}_1,\action{}_1)\vert\state)=\exp(\utility{}_1(\observationhistory{}_1,\action{}_1))$. The (joint) action space of this case consists of two (joint) actions: $\jointactionspace_{\joint}=\jointactionspace{}_1=\{\action{}^*_1, \action{}'_1\}$, where $\action{}^*_1$ is the optimal action (with expected return $2$) and $\action{}'_1$ is the suboptimal action (with expected return $1.5$). Given the probability mass function (PMF) of $\utility{}_1(\observationhistory{}_1,\action{}^*_1)$ as:
\begin{equation}
p_1(z)=\begin{cases}
1 &\text{if } z=2\\
0   &\text{otherwise},\\
\end{cases}
\notag{}
\end{equation}
and the PMF of $\utility{}_1(\observationhistory{}_1,\action{}'_1)$ as:
\begin{equation}
p_2(z)=\begin{cases}
0.5 &\text{if } z=0\\
0.5 &\text{if } z=3\\
0   &\text{otherwise}.\\
\end{cases}
\notag{}
\end{equation}
Based on this special case, we calculate the optimal actions before applying and after applying the monotonic transformation as follows:
\begin{equation}
\begin{split}
&\begin{cases}
&\mathbb{E}[\utility{}_1(\observationhistory{}_1,\action{}^*_1)]=1\cdot 2=2 \\
&\mathbb{E}[\utility{}_1(\observationhistory{}_1,\action{}'_1)]=0.5\cdot 0+0.5\cdot 3=1.5 \\
&\arg\max_{\action{}_1}\mathbb{E}[\utility{}_1(\observationhistory{}_1,\action{}_1)]=\action^*_1\\
\end{cases}\\
&\begin{cases}
&\mathbb{E}[\monotonicfunction{}(\utility{}_1(\observationhistory{}_1,\action{}^*_1)\vert\state)]=\mathbb{E}[\exp(\utility{}_1(\observationhistory{}_1,\action{}^*_1))]=e^2\approx 7.39 \\
&\mathbb{E}[\monotonicfunction{}(\utility{}_1(\observationhistory{}_1,\action{}'_1)\vert\state)]=\mathbb{E}[\exp(\utility{}_1(\observationhistory{}_1,\action{}'_1))]=0.5\cdot e^0+0.5\cdot e^3\approx 10.54 \\
&\arg\max_{\action_1}\mathbb{E}[\monotonicfunction{}(\utility{}_1(\observationhistory{}_1,\action{}_1)\vert\state)]=\action'_1\\
\end{cases}
\end{split}
\notag{}
\end{equation}
Assume, to the contrary, that \monotonicity{} for utility distributions is a sufficient condition for DIGM. By the definition of DIGM (Definition~\ref{def:distributional_igm}):
\begin{equation}
\begin{split}
&\arg\max_{\jointaction{}}\mathbb{E}[\utility{}_{\joint{}}(\jointobservationhistory{},{\jointaction{}})]=
\begin{pmatrix}
\arg\max_{\action{}_1}\mathbb{E}[\utility{}_1(\observationhistory{}_1,\action{}_1)]
\end{pmatrix}\\
\Rightarrow\ &\arg\max_{\action_1}\mathbb{E}[\monotonicfunction{}(\utility{}_1(\observationhistory{}_1,\action{}_1)\vert\state)]=\arg\max_{\action{}_1}\mathbb{E}[\utility{}_1(\observationhistory{}_1,\action{}_1)]\\
\Rightarrow\ &\action'_1=\action^*_1 \ (\Rightarrow\!\Leftarrow \text{contradiction}). \\
\end{split}
\notag{}
\end{equation}
A contradiction occurs since $\action'_1\ne\action^*_1$, showing that \monotonicity{} is not a sufficient condition for DIGM. Since there exists a case where DIGM does not hold for $\numberofagents=1$, it certainly does not hold for all $\numberofagents\in\mathbb{Z^{+}}$.
\end{proof}
The reason that the naive approach does not hold for \monotonicity{} is because the monotonic function (i.e., $\exp$) in the above example
reshapes the distributions in a non-linear manner. This may potentially cause reordering of the actions’ expected returns. For example, consider a case of two distributions, a left-skewed and a right-skewed distributions with the same expectation. After applying the exponential transformation, the expectation of the left-skewed distribution becomes larger than that of the right-skewed distribution. Such an effect only exists in non-linear transformations and does not exist in linear transformations.
\begin{theorem}
\label{thm:distributional_igm}

Given a deterministic joint action-value function $\utilityexp_{\joint}$, a stochastic joint action-value function $\utility_{\joint}$, and a factorization function $\meandecompositionfunction$ for deterministic utilities:
\begin{equation}
\utilityexp{}_{\joint{}}(\jointobservationhistory{},{\jointaction{}})=\meandecompositionfunction(\utilityexp{}_1(\observationhistory{}_1,\action{}_1), ..., \utilityexp{}_{\numberofagents{}}(\observationhistory{}_{\numberofagents{}},\action{}_{\numberofagents{}})\vert\state),
\notag{}
\end{equation}
such that $[\utilityexp{}_\agentcounter{}]_{\agentcounter{}\in\agentspace{}}$ satisfy IGM for $\utilityexp{}_{\joint{}}$ under $\jointobservationhistory{}$, the following distributional factorization:
\begin{equation}
\utility{}_{\joint{}}(\jointobservationhistory{},{\jointaction{}})\eqd\meandecompositionfunction(\utility{}_1(\observationhistory{}_1,\action{}_1), ..., \utility{}_{\numberofagents{}}(\observationhistory{}_{\numberofagents{}},\action{}_{\numberofagents{}})\vert\state).
\notag{}
\end{equation}
\end{theorem}
is insufficient to guarantee that $[\utility{}_\agentcounter{}]_{\agentcounter{}\in\agentspace{}}$ satisfy \digm{} for $\utility{}_{\joint{}}$ under $\jointobservationhistory{}$.

\begin{proof}
A contradiction is provided by Proposition~\ref{prop:distributional_monotonicity}.
\end{proof}
According to the observations in Theorem~\ref{thm:distributional_igm}, an alternative strategy (other than the naive approach) for modifying the factorization functions is necessary to meet the requirements of satisfying \digm{} for stochastic joint distributions.


\subsection{Mean-Shape Decomposition and the DFAC Framework}
\label{subsec:methodology_the_proposed_dfac_framework}

We propose Mean-Shape Decomposition and the DFAC framework to ensure that \digm{} is satisfied for utility distributions. We first define Mean-Shape Decomposition as follows:
\begin{definition}[Mean-Shape Decomposition]
\label{def:mean_shape_decomposition}

For any given random variable $Z$, there exists a unique decomposition defined as follows:
\begin{equation}
\begin{split}
Z &= \mathbb{E}[Z]+(Z-\mathbb{E}[Z]) \\
&= Z_{\mathrm{mean}}+Z_{\mathrm{shape}}\ ,
\end{split}
\notag{}
\end{equation}
where $\mathrm{Var}(Z_{\mathrm{mean}})=0$ and $\mathbb{E}[Z_{\mathrm{shape}}]=0$. In this paper, this is called the Mean-Shape Decomposition of $Z$.

\end{definition}
Based on Mean-Shape Decomposition, we propose DFAC to decompose a joint return distribution $\utility{}_{\joint{}}$ into its deterministic part $\utility{}_{\text{mean}}$ (i.e., the expected value) and stochastic part $\utility{}_{\text{shape}}$ (i.e., the higher moments). The two components $\utility{}_{\text{mean}}$ and $\utility{}_{\text{shape}}$ are approximated by two different functions $\meandecompositionfunction$ and $\shapedecompositionfunction$, respectively. The factorization function $\meandecompositionfunction$ is responsible for factorizing the expectation of $\utility{}_{\joint{}}$, while the shape function $\shapedecompositionfunction$ is utilized to factorize the shape of $\utility{}_{\joint{}}$. Since the main objective of modeling the return distribution is to assist non-linear approximation of the expectation of $\utility{}_{\joint{}}$ \citep{Lyle2019Comparative,BDR2022} rather than accurately model the shape of $\utility{}_{\joint{}}$, $\shapedecompositionfunction$ is allowed to roughly factorize the shape of $\utility{}_{\joint{}}$.
\begin{theorem}[DFAC Theorem]
\label{thm:dfac}

Consider a deterministic joint action-value function $\utilityexp_{\joint}$, a stochastic joint action-value function $\utility_{\joint}$, and a factorization function $\meandecompositionfunction$ for deterministic utilities:
\begin{equation}
\utilityexp{}_{\joint{}}(\jointobservationhistory{},{\jointaction{}})=\meandecompositionfunction(\utilityexp{}_1(\observationhistory{}_1,\action{}_1), ..., \utilityexp{}_{\numberofagents{}}(\observationhistory{}_{\numberofagents{}},\action{}_{\numberofagents{}})\vert\state),
\notag{}
\end{equation}
such that $[\utilityexp{}_\agentcounter{}]_{\agentcounter{}\in\agentspace{}}$ satisfy IGM for $\utilityexp{}_{\joint{}}$ under $\jointobservationhistory{}$. By Mean-Shape Decomposition, the following distributional factorization:
\begin{equation}
\begin{split}
\utility{}_{\joint{}}(\jointobservationhistory{},{\jointaction{}}) &= \mathbb{E}[\utility{}_{\joint{}}(\jointobservationhistory{},{\jointaction{}})]+(\utility{}_{\joint{}}(\jointobservationhistory{},{\jointaction{}})-\mathbb{E}[\utility{}_{\joint{}}(\jointobservationhistory{},{\jointaction{}})]) \\
&= \utility{}_{\mathrm{mean}}(\jointobservationhistory{},{\jointaction{}})+\utility{}_{\mathrm{shape}}(\jointobservationhistory{},{\jointaction{}})\\
&\eqd \meandecompositionfunction{}(\utilityexp_1(\observationhistory{}_1,\action{}_1),...,\utilityexp_{\numberofagents{}}(\observationhistory{}_{\numberofagents{}},\action{}_{\numberofagents{}})\vert\state) + \shapedecompositionfunction{}(\utility_1(\observationhistory{}_1,\action{}_1),...,\utility_{\numberofagents{}}(\observationhistory{}_{\numberofagents{}},\action{}_{\numberofagents{}})\vert\state),
\end{split}
\notag{}
\end{equation}
is sufficient to guarantee that $[\utility{}_\agentcounter{}]_{\agentcounter{}\in\agentspace{}}$ satisfy \digm{} for $\utility{}_{\joint{}}$ under $\jointobservationhistory{}$, where $\mathbb{E}[\shapedecompositionfunction{}(\utility_1,...,\utility_{\numberofagents{}}\vert\state)]=0$ for all $[\utility{}_\agentcounter{}]_{\agentcounter{}\in\agentspace{}}$.
\end{theorem}
\begin{proof}
Based on Mean-Shape decomposition:
\begin{equation}
\begin{aligned}
&\arg\max_\jointaction{}\{ \mathbb{E}[\utility{}_{\joint{}}(\jointobservationhistory{},\jointaction{})]\}\\
=\ &\arg\max_\jointaction{}\{ \mathbb{E}[\utility{}_{\mathrm{mean}}(\jointobservationhistory{},\jointaction{})+\utility{}_{\mathrm{shape}}(\jointobservationhistory{},\jointaction{})]\}\\
=\ &\arg\max_\jointaction{} \{\mathbb{E}[\utility{}_{\mathrm{mean}}(\jointobservationhistory{},\jointaction{})]+\mathbb{E}[\utility{}_{\mathrm{shape}}(\jointobservationhistory{},\jointaction{})]\}\\
=\ &\arg\max_\jointaction{} \{\mathbb{E}[\meandecompositionfunction(\utilityexp{}_1(\observationhistory{}_1,\action{}_1), ..., \utilityexp{}_{\numberofagents{}}(\observationhistory{}_{\numberofagents{}},\action{}_{\numberofagents{}})\vert\state)]+\mathbb{E}[\shapedecompositionfunction(\utility{}_1(\observationhistory{}_1,\action{}_1), ..., \utility{}_{\numberofagents{}}(\observationhistory{}_{\numberofagents{}},\action{}_{\numberofagents{}})\vert\state)]\}\\
=\ &\arg\max_\jointaction{} \{\meandecompositionfunction(\utilityexp{}_1(\observationhistory{}_1,\action{}_1), ..., \utilityexp{}_{\numberofagents{}}(\observationhistory{}_{\numberofagents{}},\action{}_{\numberofagents{}})\vert\state)+0\}\\
=\ &\arg\max_\jointaction{} \{\meandecompositionfunction(\utilityexp{}_1(\observationhistory{}_1,\action{}_1), ..., \utilityexp{}_{\numberofagents{}}(\observationhistory{}_{\numberofagents{}},\action{}_{\numberofagents{}})\vert\state)\}\\
=\ &\begin{pmatrix}
\arg\max_{\action{}_1} \utilityexp{}_1(\observationhistory{}_1,\action{}_1)\\
\vdots \\
\arg\max_{\action{}_\numberofagents{}} \utilityexp{}_\numberofagents{}(\observationhistory{}_\numberofagents{},\action{}_\numberofagents{})
\end{pmatrix}\\
\Rightarrow\ &\arg\max_\jointaction{} \mathbb{E}[\utility{}_{\joint{}}(\jointobservationhistory{},\jointaction{})]
=\ \begin{pmatrix}
\arg\max_{\action{}_1} \mathbb{E}[\utility{}_1(\observationhistory{}_1,\action{}_1)]\\
\vdots \\
\arg\max_{\action{}_\numberofagents{}} \mathbb{E}[\utility{}_\numberofagents{}(\observationhistory{}_\numberofagents{},\action{}_\numberofagents{})]
\end{pmatrix}.\\
\end{aligned}
\notag{}
\end{equation}
The above derivation demonstrates that $[\utility{}_\agentcounter{}]_{\agentcounter{}\in\agentspace{}}$ satisfy \digm{} for $\utility{}_{\joint{}}$ under $\jointobservationhistory{}$.
\end{proof}
Theorem~\ref{thm:dfac} reveals that the choice of $\meandecompositionfunction$ determines whether DIGM holds, regardless of the choice of $\shapedecompositionfunction$, as long as $\mathbb{E}[\shapedecompositionfunction{}(\utility_1,...,\utility_{\numberofagents{}}\vert\state)]=0$ for all $[\utility{}_\agentcounter{}]_{\agentcounter{}\in\agentspace{}}$. Under this setting, any factorization function of deterministic variables can be extended to a factorization function of stochastic variables. Such a decomposition enables approximation of joint distributions for all factorizable tasks under appropriate choices of $\meandecompositionfunction{}$ and $\shapedecompositionfunction{}$. The methods extended by DFAC is called the DFAC variants, which have the same expressiveness as the original unextended versions.


\subsection{Practical Implementation Choices of the Shape Function}
\label{subsec:methodology_a_practical_implementation_of_dfac}

In this section, we discuss the implementation choices of the shape function based on two representative distributional algorithms: C51 and IQN. Theoretically, the shape function can be arbitrarily complex as long as the output has a zero mean. In practice, however, such a complicated network is unnecessary, since learning the shape of distribution is only an auxiliary task to facilitate representation learning, as mentioned in Section~\ref{subsec:background_distributional_rl}. For the sake of simplicity, we choose the shape function to be additive and non-learnable: $Z_\mathrm{shape}=\shapedecompositionfunction(\utility{}_1(\observationhistory{}_1,\action{}_1), ..., \utility{}_{\numberofagents{}}(\observationhistory{}_{\numberofagents{}},\action{}_{\numberofagents{}})\vert\state)=\sum_{\agentcounter\in\agentspace}(Z_\agentcounter-\mathbb{E}[Z_\agentcounter])$.

\subsubsection{Shape Function for C51}
\label{subsec:shape_function_for_c51}
In this section, we demonstrate the implementation of the shape function for C51, which models the utilities as PMFs.
We denote the PMFs of the utilities as $[f_\agentcounter]_{\agentcounter\in\agentspace}$, and the PMF of the joint return as $f$.
For the sake of simplicity, we first consider the case for two agents ($\numberofagents{}=2$).
To exactly factorize the total return ($Z=Z_1+Z_2$), the PMF $f$ of the joint return $Z$ can be expressed as:
\begin{equation}
\begin{split}
f(x)&=\Pr(Z=x) \\
&=\Pr(Z_1+Z_2=x) \\
&=\sum_{y\in\mathbb{R}} \Pr(Z_1=y)\Pr(Z_2=x-y|Z_1=y).
\end{split}
\label{eq:exact_decomposition}
\end{equation}
Since the above convolution requires the calculation of conditional probability, it is computationally intractable. A common solution to it is to assume $Z_1$ and $Z_2$ to be mutually independent, which is a suitable assumption for decentralized utilities. This approach can be viewed as an approximation of the exact factorization method (Eq~(\ref{eq:exact_decomposition})).
The approximation can be expressed as the following:
\begin{equation}
\begin{split}
f(x)&=\sum_{y\in\mathbb{R}} \Pr(Z_1=y)\Pr(Z_2=x-y|Z_1=y) \\
&\approx\sum_{y\in\mathbb{R}} \Pr(Z_1=y)\Pr(Z_2=x-y) \\
&=(f_1*f_2)(x),
\label{eq:decomposition}
\end{split}
\end{equation}
where the star symbol ($*$) denotes the convolution operation. Eq~(\ref{eq:decomposition}) is tractable and allows C51 to model the PMF of each agent by a categorical distribution with $n$ atoms (i.e., categories). The time complexity of deriving $f$ for $\numberofagents{}$ agents is:
\begin{equation}
    n\cdot n+n^2\cdot n+...+n^{\numberofagents{}-1}\cdot n=\sum^{\numberofagents{}-1}_{k=1}n^{k+1}=O(n^\numberofagents{}).
    \label{eq:projection}
\end{equation}
However, the time complexity grows exponentially as the number of the agents $\numberofagents{}$ increases. This becomes the training bottleneck and remains computationally infeasible when $\numberofagents{}$ is scaled to a large number. The computation cost can be reduced by performing an additional $O(n)$ heuristic projection after each convolution operation~\citep{Lin2019DR-DRL}. The time complexity is then simplified to the following:
\begin{equation}
    (n\cdot n + O(n))\cdot(\numberofagents{}-1)=O(\numberofagents{}n^2).
\end{equation}
The additional heuristic projections enables the decomposition process to be computationally feasible. However, such projections may induce an increase in variance for $Z$~\citep{Rowland2019ER-DQN}. If the return distribution of each agent shares the same support, the above $O(\numberofagents{}n^2)$ complexity can be improved by Fast Fourier Transform (FFT) convolutions. The technique is based on the convolution theorem, and is able to further reduce the time complexity to:
\begin{equation}
    (n\log n + n + n\log n + O(n))\cdot(\numberofagents{}-1)=O(\numberofagents{}n\log n).
\end{equation}
\subsubsection{Shape Function for IQN}

Distributional RL methods such as IQN~\citep{Dabney2018IQN} do not explicitly model a distribution. They only implicitly model a distribution, and hence, the previously mentioned decomposition methods for C51 are not directly applicable. A potential solution is to add constraints on IQN to make the modeled distribution explicit~\citep{Yang2019FQF}, and iteratively apply convolution operations. To maintain the flexibility of IQN and to further reduce the time complexity, instead of imposing additional constraints, we propose to use quantile mixture to approximately factorize the distribution. In the following theorem, we prove that such a method is a valid factorization.
\begin{theorem}
\label{thm:sum_of_rv}
Given a quantile mixture:
\begin{equation}
\quantilefunction(\quantile{})=\sum^{\numberofagents{}}_{k=1}\modelparameter{}_{\agentcounter{}} \quantilefunction_{\agentcounter{}}(\quantile{}),
\notag{}
\end{equation}
with $\numberofagents{}$ components $[\quantilefunction_{\agentcounter{}}]_{\agentcounter{}\in\agentspace{}}$, non-negative model parameters $[\modelparameter{}_{\agentcounter{}}]_{\agentcounter{}\in\agentspace{}}$, and $\quantile\in[0,1]$. There exist random variables $\utility{}$ and $[\utility_{\agentcounter{}}]_{\agentcounter{}\in\agentspace{}}$ derived from the quantile functions $\quantilefunction$ and $[\quantilefunction_{\agentcounter{}}]_{\agentcounter{}\in\agentspace{}}$, respectively, with the following relationship:
\begin{equation*}
\utility{}=\sum_{\agentcounter\in\agentspace{}}\modelparameter{}_{\agentcounter{}} \utility{}_{\agentcounter{}}.
\end{equation*}

\end{theorem}
\begin{proof}
The proof is outlined as follows. First, we provide the notational details of the random variables. Next, we define the explicit forms of $Z$ and $[Z_k]_{\agentcounter\in\agentspace}$. Finally, we prove the relationship $Z=\sum_{\agentcounter\in\agentspace}\modelparameter{}_{\agentcounter{}}\utility_{\agentcounter{}}$.

The domain of $Z$ and $[Z_k]_{\agentcounter\in\agentspace}$ is $\Omega$, where $\Omega$ is the sample space of the probability space $(\Omega,\mathcal{F}, P)$.
The codomain of $Z$ and $[Z_k]_{\agentcounter\in\agentspace}$ is $\mathbb{R}$, where $\mathbb{R}$ is the sample space of the probability spaces $(\mathbb{R}, \mathcal{B}, \mu)$ and $[(\mathbb{R}, \mathcal{B}, \mu_k)]_{\agentcounter\in\agentspace}$.
Based on the definition of random variables in the measure theory, $(\Omega,\mathcal{F}, P)$, $(\mathbb{R}, \mathcal{B}, \mu)$, and $[(\mathbb{R}, \mathcal{B}, \mu_k)]_{\agentcounter\in\agentspace}$ are defined as follows:
\begin{equation*}
\begin{split}
& \begin{cases}
\Omega=[0,1], \text{ is the sample space},\\
\mathcal{F}, \text{ is the Borel $\sigma$-algebra over $\Omega$},\\
P:\mathcal{F}\rightarrow [0,1], \text{ is a probability measure on the measurable space $(\Omega,\mathcal{F})$}.
\end{cases}\\
& \begin{cases}
\mathbb{R}=(-\infty,\infty), \text{ is the real line},\\
\mathcal{B}, \text{ is the Borel $\sigma$-algebra over $\mathbb{R}$},\\
\mu:\mathcal{B}\rightarrow [0,1], \text{ is a probability measure on the measurable space $(\mathbb{R}, \mathcal{B})$},\\
[\mu_k:\mathcal{B}\rightarrow [0,1]]_{\agentcounter\in\agentspace}, \text{ are probability measures on the measurable space $(\mathbb{R}, \mathcal{B})$}.
\end{cases}
\end{split}
\end{equation*}
The random variables $Z$ and $[Z_k]_{\agentcounter\in\agentspace}$ are defined as functions based on the probability spaces:
\begin{equation*}
\begin{cases}
Z:\Omega\rightarrow\mathbb{R}, \text{ with its CDF defined as $F(x)=\mu(-\infty,x]$},\\
[Z_k:\Omega\rightarrow\mathbb{R}]_{\agentcounter\in\agentspace}, \text{ with their CDFs defined as $F_k(x)=\mu_k(-\infty,x], \forall \agentcounter\in\agentspace$}.\\
\end{cases}
\end{equation*}
We can then construct the random variables $Z$ and $[\utility_{\agentcounter{}}]_{\agentcounter{}\in\agentspace{}}$ based on the components of the quantile mixture:
\begin{equation*}
\begin{cases}
\utility(\quantile)=\quantilefunction(\quantile), \forall \quantile\in\Omega, \\
\utility_\agentcounter{}(\quantile)=\quantilefunction_{\agentcounter{}}(\quantile), \forall \quantile\in\Omega, \forall \agentcounter\in\agentspace. \\
\end{cases}
\end{equation*}
\begin{equation*}
\end{equation*}
With the explicit forms of $Z$ and $[Z_k]_{\agentcounter\in\agentspace}$, their relationship can be proved as follows:
\begin{align*}
& \quantilefunction(\quantile{})=\sum^{\numberofagents{}}_{k=1}\modelparameter{}_{\agentcounter{}} \quantilefunction_{\agentcounter{}}(\quantile{}) \\
\Rightarrow\ & Z(\quantile)=\sum_{\agentcounter\in\agentspace}\modelparameter{}_{\agentcounter{}}\utility_{\agentcounter{}}(\quantile) \\
\Rightarrow\ & Z=\sum_{\agentcounter\in\agentspace}\modelparameter{}_{\agentcounter{}}\utility_{\agentcounter{}}
\end{align*}
\end{proof}
In the optimization process, assume that all components $[\quantilefunction_{\agentcounter{}}]_{\agentcounter{}\in\agentspace{}}$ in the quantile mixture sample $n$ quantiles, the total time complexity of summing the components is $O(\numberofagents{}n)$, which grows linearly with respect to $\numberofagents{}$ and $n$. The implementation is much easier and efficient when compared with the other methods described in the previous section. This suggests that quantile mixture is unlikely to become the training bottleneck, even if the values of $\numberofagents{}$ and $n$ are both scaled to large numbers. It is worth noting that these two shape functions (i.e., convolution and quantile mixture) make different assumptions for modeling the joint return distribution.

This approach, however, trades off the exact factorization of the shape of the distribution for linear time complexity. Therefore, such a technique should only be used when the errors of capturing the shapes of the factorized distributions are acceptable. This technique is especially suitable for value function factorization methods since capturing the shapes of the factorized distributions is an auxiliary task.

\begin{figure*}[t]
\includegraphics[width=\linewidth]{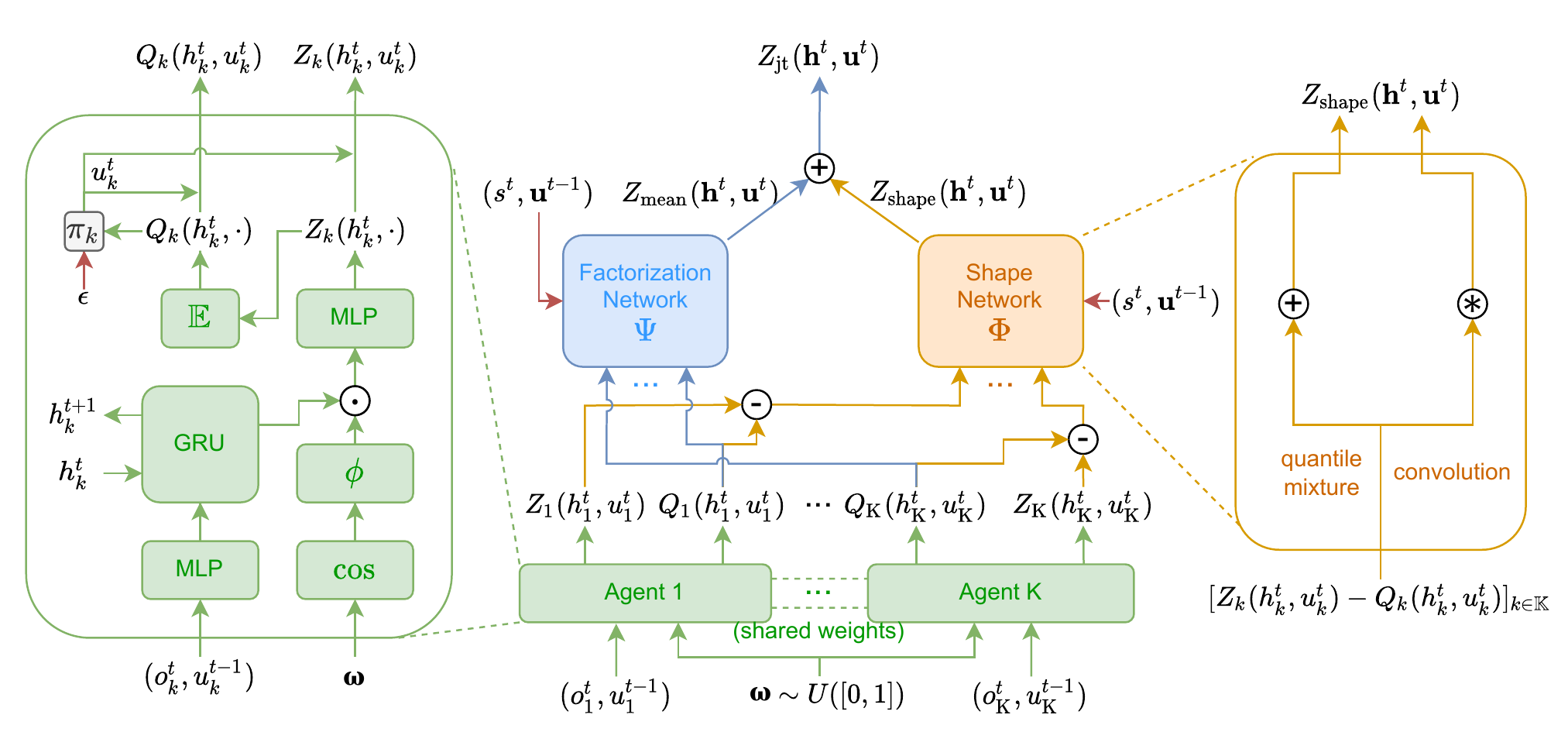}
\caption{The DFAC framework consists of a factorization function $\meandecompositionfunction{}$ and a shape function $\shapedecompositionfunction$ for decomposing the deterministic part $\utility{}_{\text{mean}}$ (i.e., $\utilityexp{}_{\joint{}}$) and the stochastic part $\utility{}_{\text{shape}}$ of the total return distribution $\utility{}_{\joint{}}$, as described in Theorem~\ref{thm:dfac}. If the shape function is a convolution, the network of the agents is implemented as C51. If the shape function is a quantile mixture, the network of the agents is implemented as IQN.
}
\label{fig:dfac}
\end{figure*}


\subsection{DFAC Variant of VDN, QMIX, and QPLEX}
\label{subsec:methodology_distributional_variant_of_vdn_and_qmix}

In order to validate the proposed DFAC framework, we next discuss the DFAC variants of three representative factorization methods: VDN, QMIX, and QPLEX.
For simplicity, these DFAC variants assume the shape function to be additive and non-learnable, and uses the quantile mixture technique instead of convolution.
The DFAC variant of VDN is named \ddn{}, which is expressed as follows:
\begin{equation}
\utility_\joint=\sum_{\agentcounter\in\agentspace{}}\utilityexp_{\agentcounter}+\sum_{\agentcounter\in\agentspace{}}(\utility_{\agentcounter}-\utilityexp_{\agentcounter}).
\label{eq:ddn}
\end{equation}
Eq~(\ref{eq:ddn}) combines the factorization function of VDN (i.e., $\meandecompositionfunction=\sum_{\agentcounter\in\agentspace{}}\utilityexp_{\agentcounter}$) and the additive shape function (i.e.,  $\shapedecompositionfunction=\sum_{\agentcounter\in\agentspace{}}(\utility_{\agentcounter}-\utilityexp_{\agentcounter})$).
In addition, the DFAC variant of QMIX is named DMIX, which is defined as the following:
\begin{equation}
\utility_\joint=\monotonicfunction{}(\utilityexp_1,...,\utilityexp_\numberofagents\vert\state)+\sum_{\agentcounter\in\agentspace{}}(\utility_{\agentcounter}-\utilityexp_{\agentcounter}).
\label{eq:dmix}
\end{equation}
Eq~(\ref{eq:dmix}) combines the factorization function of QMIX (i.e., $\meandecompositionfunction=\monotonicfunction{}(\utilityexp_1,...,\utilityexp_\numberofagents\vert\state)$), and the additive shape function (i.e., $\shapedecompositionfunction=\sum_{\agentcounter\in\agentspace{}}(\utility_{\agentcounter}-\utilityexp_{\agentcounter})$). Lastly, the DFAC variant of QPLEX is named DPLEX, which is formulated as follows:
\begin{equation}
\utilityexp{}_{\joint{}}(\jointobservationhistory{},\jointaction{}) = \sum^{\numberofagents{}}_{\agentcounter{}=1} \utilitystate_\agentcounter{}(\jointobservationhistory{}) + \sum^{\numberofagents{}}_{\agentcounter{}=1}\lambda_\agentcounter{}(\jointobservationhistory{},\jointaction{}) \utilityadv_\agentcounter{}(\jointobservationhistory{},\action{}_\agentcounter{})+\sum_{\agentcounter\in\agentspace{}}(\utility_{\agentcounter}-\utilityexp_{\agentcounter}).
\label{eq:dplex}
\end{equation}
Eq~(\ref{eq:dplex})  is equivalent to summing the factorization function of QPLEX (i.e., $\meandecompositionfunction=\sum^{\numberofagents{}}_{\agentcounter{}=1} \utilitystate_\agentcounter{}(\jointobservationhistory{}) + \sum^{\numberofagents{}}_{\agentcounter{}=1}\lambda_\agentcounter{}(\jointobservationhistory{},\jointaction{}) \utilityadv_\agentcounter{}(\jointobservationhistory{},\action{}_\agentcounter{})$), and the additive shape function (i.e., $\shapedecompositionfunction=\sum_{\agentcounter\in\agentspace{}}(\utility_{\agentcounter}-\utilityexp_{\agentcounter})$).
These DFAC variants can be illustrated in the form of Fig.~\ref{fig:dfac} with different factorization functions.

\section{A Stochastic Matrix Game}
\label{sec:distributional_2_step_game}

In the previous expected value function factorization methods (e.g., VDN, QMIX, QPLEX, etc.), the factorization is achieved by modeling $\utilityexp{}_{\joint{}}$ and $[\utilityexp{}_{\agentcounter{}}]_{\agentcounter{}\in\agentspace{}}$ as deterministic variables, overlooking the information of higher moments in the full return distributions $\utility{}_{\joint{}}$ and $[\utility{}_{\agentcounter{}}]_{\agentcounter{}\in\agentspace{}}$.
In order to demonstrate DFAC's ability of factorization, we begin with a toy example modified from~\cite{Wang2020QPLEX} to show that DFAC is able to approximate the true return distributions, and factorize the mean and variance of the approximated total return $\utility{}_{\joint{}}$ into utilities $[\utility{}_{\agentcounter{}}]_{\agentcounter{}\in\agentspace{}}$.


\subsection{Experimental Setup}
\label{subsec:matrix_game_setup}

In this section, we describe the environment and methods, the architecture of the policy network, the architecture of the factorization network, and finally the training and execution details.

\textbf{Environment and Methods.}
Table~\ref{table:matrix-game} illustrates a fully cooperative matrix game with stochastic rewards for two agents. The game only contains a single state, and each episode contains only a single step. In each episode, both agents choose an action and receive a global reward sampled from a normal distribution $\mathcal{N}(\mu,\sigma^2)$ with mean $\mu$ and standard deviation $\sigma$.
In this matrix game, we perform experiments on nine value function factorization methods: VDN, QMIX, and QPLEX are the baselines; DDN, DMIX, and DPLEX are the DFAC variants using IQN and quantile mixture; DDN-C51, DMIX-C51, and DPLEX-C51 correspond to the DFAC variants that use C51 and convolution. These methods are tuned independently and may have different hyperparameters.

\textbf{Architecture of the Policy Network.}
Each agent's policy network is implemented as an artificial neural network (ANN). QPLEX is implemented as a single-layered ANN with its hidden layer comprised of $32$ units. DPLEX and DPLEX-C51 are implemented as two-layered ANNs comprised of $64$ units and $32$ units, with a ReLU nonlinearity at the end of the first layer. The other methods follow a similar architecture used by DPLEX and DPLEX-C51, where the second layer is changed to $512$ units.

For DDN and DMIX, we optimize the IQNs of the agents with $\numberofquantiles{}=\numberofquantilesamples{}=32$ quantile samples, where each of them is first encoded into a $64$-dimensional intermediate embedding and then projected to a $512$-dimensional quantile embedding by a single hidden layer. As for DPLEX, we optimize the IQNs with $\numberofquantiles{}=\numberofquantilesamples{}=512$ quantile samples, where each of them is first encoded  into a $256$-dimensional intermediate embedding and then projected to a $32$-dimensional quantile embedding by a single hidden layer.

For DDN-C51 and DMIX-C51, we optimize the C51 networks of the agents with $51$ atoms uniformly distributed in $[-20,20]$.

\textbf{Architecture of the Factorization Network.}
For \dmix{} and its variants, we use a single-layered mixing network with eight units. As for QPLEX and its variants, we use a single-layered monotonic network with $16$ units, and an attention network $\lambda_\agentcounter{}$ (i.e., in Eq~(\ref{eq:qplex})) with ten attention heads. Each of the attention heads utilizes a three-layered network with $64$ units in the intermediate embeddings for its key, value, and query extractors.

\textbf{Training and Execution.}
During training, each agent performs independent $\epsilon$-greedy action selection, with full exploration (i.e., $\epsilon=1$). The replay buffer contains experiences of the latest $2,000$ episodes, from which we uniformly sample a batch of $2,048$ samples when training QPLEX and its variants, and a batch of $512$ samples for training the other six methods. The target network is updated every $100$ episodes. The optimizer is set to Adam, in which its learning rate is set to $1 \times 10^{-3}$ for QPLEX and QPLEX-C51, and $1 \times 10^{-4}$ for the rest of the methods. We train each of the methods for $20,000$ episodes. All of the agent networks share the same parameters. The agents are differentiated by a one-hot encoded agent index (i.e., $[1\ 0]^T$ for agent $1$ and $[0\ 1]^T$ for agent $2$) concatenated to their observations.

\begin{figure*}[t]
\begin{tabular}{cc}
\begin{minipage}{0.47\textwidth}
\small
\setlength{\extrarowheight}{3pt}
\captionof{table}{The payoff matrix of the stochastic matrix game. Each agent performs an action from $\{A,B,C\}$, with a subscript denoting the agent's index. The global rewards are sampled from the normal distributions in the payoff matrix below.}\vspace{-0.3cm}\hspace{-0.5cm}
\begin{tabular}{cc|*{3}{>{\centering\arraybackslash}p{.22\linewidth}|}}
	& \multicolumn{1}{c}{} & \multicolumn{3}{c}{\bb{Agent $2$}} \\
	& \multicolumn{1}{c}{} & \multicolumn{1}{c}{\bb{$A_2$}} & \multicolumn{1}{c}{\bb{$B_2$}} & \multicolumn{1}{c}{\bb{$C_2$}} \\ 
	\cline{3-5}
    \multirow{2}{*}{\rotatebox[origin=c]{90}{\cc{Agent $1$}}} & \cc{$A_1$} & $\mathcal{N}(8,8)$ & $\mathcal{N}(-12,6)$ & $\mathcal{N}(-12,4)$ \\ \cline{3-5}
    & \cc{$B_1$} & $\mathcal{N}(-12,6)$ & $\mathcal{N}(6,4)$ & $\mathcal{N}(0,2)$  \\\cline{3-5}
    & \cc{$C_1$} & $\mathcal{N}(-12,4)$ & $\mathcal{N}(0,2)$ & $\mathcal{N}(6,0)$  \\\cline{3-5}
\end{tabular}
\label{table:matrix-game}
\end{minipage}
\begin{tabular}{cc}
\begin{minipage}{0.45\textwidth}
\footnotesize
\captionof{table}{The value approximation errors and the returns of different methods.}
\begin{tabular}{l|c|c|c}
\toprule
Method & Q-dist & W-dist & Return \\
\midrule
VDN & 6.91 & 7.01 & 6 \\
DDN & 6.50 & 6.50 & 6 \\
DDN-C51 & 6.01 & 6.17 & 6 \\
\hline
QMIX & 4.86 & 4.92 & 6 \\
DMIX & 3.76 & 4.58 & 6 \\
DMIX-C51 & 4.96 & 5.29 & 6 \\
\hline
QPLEX & 0.10 & 1.47 & \textbf{8} \\
DPLEX & \textbf{0.06} & \textbf{0.24} & \textbf{8} \\
DPLEX-C51 & 0.33 & 0.47 & 6 \\
\bottomrule
\end{tabular}
\label{table:matrix-game-results}
\end{minipage}
\end{tabular}
\end{tabular}
\end{figure*}


\subsection{Experimental Results}

Each column of Table~\ref{table:matrix-game-results} presents a key metric of the learned policy after convergence. \textit{Q-dist} is defined as $\mathbb{E}[\vert\utilityexp_\joint - \utilityexp^*_\joint\vert]$, which is the averaged absolute distance between the approximated Q-value $\utilityexp_\joint$ and the true Q-value $\utilityexp^*_\joint$ under the optimal policy across all joint actions. \textit{W-dist} represents $\mathbb{E}[W_1(\utility_\joint, \utility^*_\joint)]$, which is the averaged 1-Wasserstein distance between the approximated return $\utility_\joint$ and the true return $\utility^*_\joint$ under the optimal policy. \textit{Return} reports the expectation of the returns following the learned policy.
Each row of Table~\ref{table:matrix-game-results} corresponds to one of the methods mentioned in Section~\ref{subsec:matrix_game_setup}.

Based on the results, we observed that the methods with more expressiveness tend to have lower values of \textit{Q-dist}. More specifically, QPLEX and its variants are able to model the expected return more precisely when compared to VDN, QMIX, and their variants.
Moreover, the DFAC variants tend to have lower values of \textit{W-dist}, reflecting DFAC's capability of modeling the shapes of return distributions better than the baselines.
Furthermore, QPLEX and DPLEX can learn the optimal return for this game (i.e., \textit{Return} $=8$), while VDN, QMIX, and their variants are unable to, which aligns with the results presented in~\cite{Wang2020QPLEX}.
Unfortunately, DPLEX-C51 fails to learn the optimal return, even though it has the same theoretical expressiveness as QPLEX. We believe that this result is potentially due to the combination of C51's shortcomings described in Section~\ref{subsec:background_distributional_rl} and the convolutional operations used in the shape function for C51 mentioned in Section~\ref{subsec:shape_function_for_c51}.
\begin{figure*}[t]
\begin{tabular}{cc}
\begin{minipage}{0.7\textwidth} 
\includegraphics[width=\linewidth]{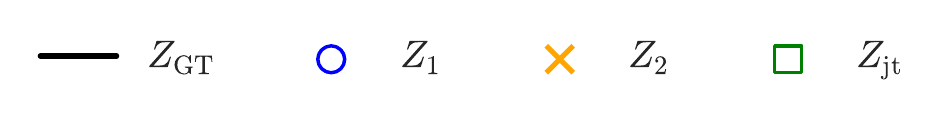}\vspace{-2em}
\label{fig:matrix-game-results-legend}
\end{minipage} \\
\begin{minipage}{0.3\textwidth} 
\includegraphics[width=\linewidth]{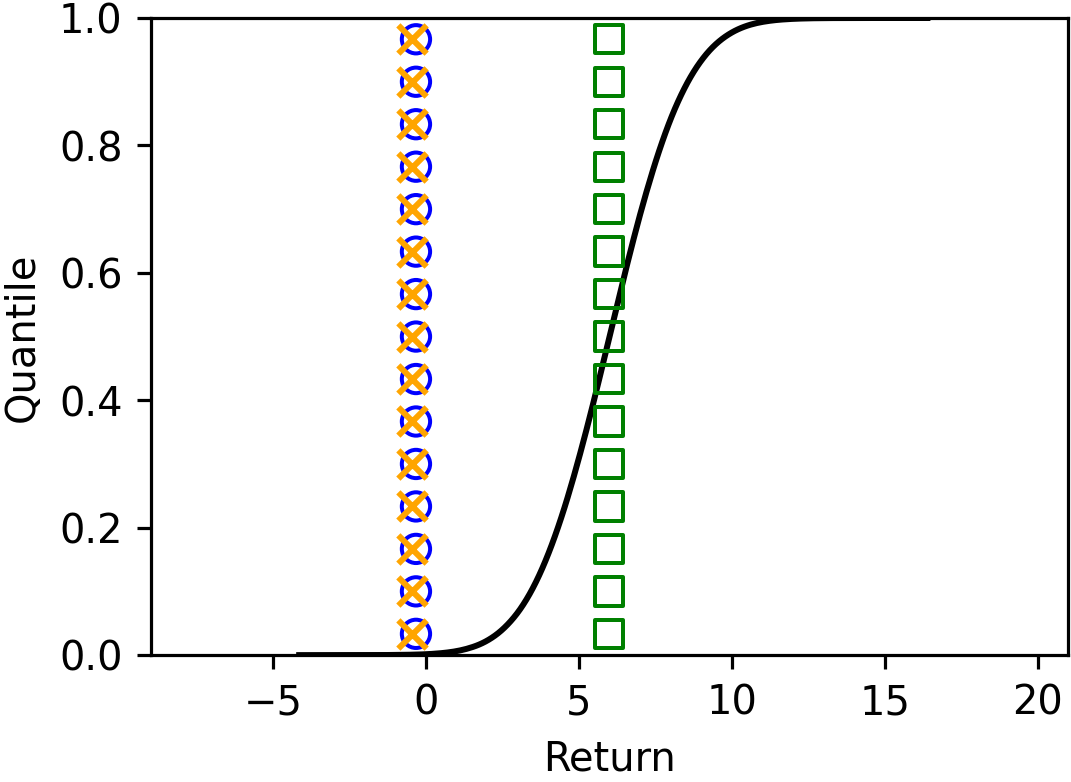}
\label{fig:matrix-game-results-qplex-1}
\vspace{-1.5em}
\subcaption{QPLEX}
\end{minipage}
\begin{minipage}{0.3\textwidth} 
\includegraphics[width=\linewidth]{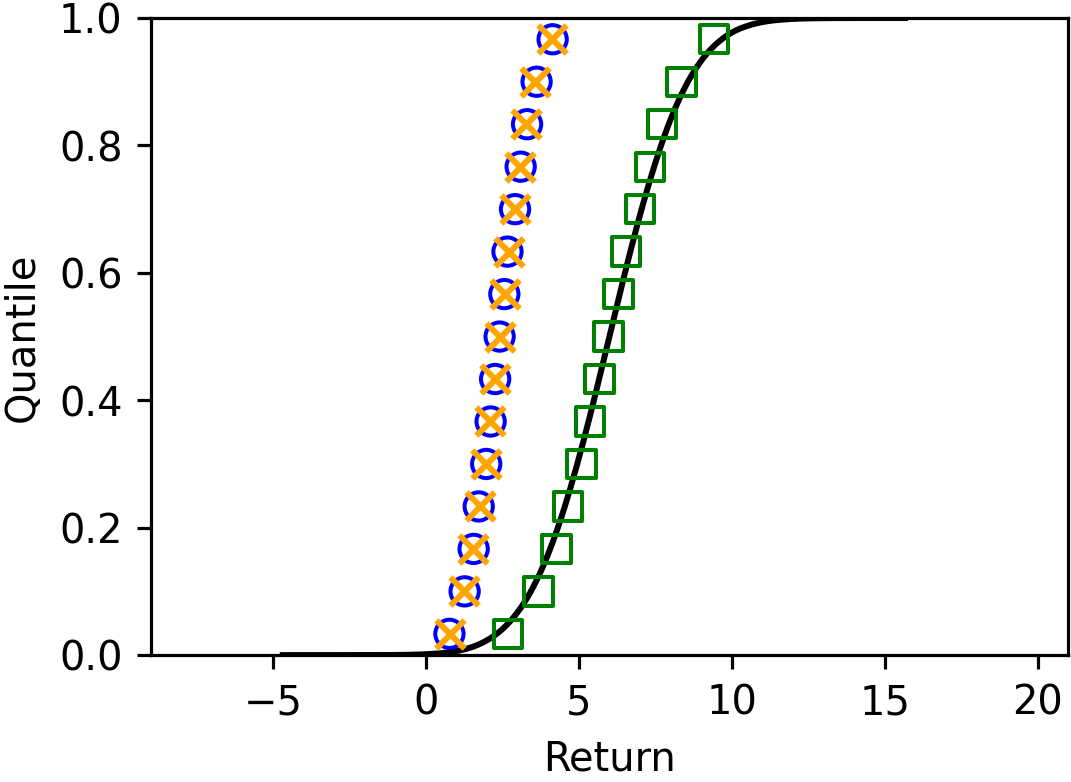}
\label{fig:matrix-game-results-dplex-1}
\vspace{-1.5em}
\subcaption{DPLEX}
\end{minipage}
\begin{minipage}{0.3\textwidth} 
\includegraphics[width=\linewidth]{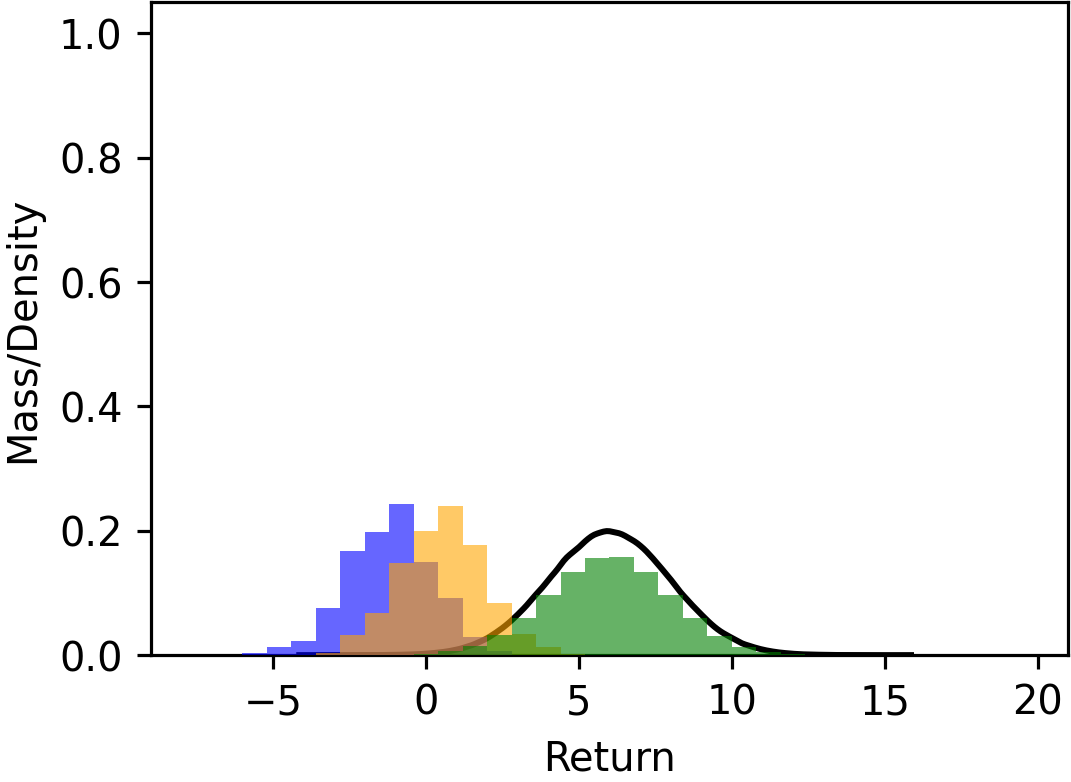}
\label{fig:matrix-game-results-dplex-c51-1}
\vspace{-1.5em}
\subcaption{DPLEX-C51}
\end{minipage}
\end{tabular}
\caption{The learned factorization of the joint values when the joint action $\langle B_1, B_2\rangle$ is selected. The ground truth ($Z_\mathrm{GT}$) of the joint value is the stochastic return $\mathcal{N}(6,4)$.}
\label{fig:matrix-game-results-1}

\vspace{1em}

\begin{tabular}{cc}
\begin{minipage}{0.3\textwidth} 
\includegraphics[width=\linewidth]{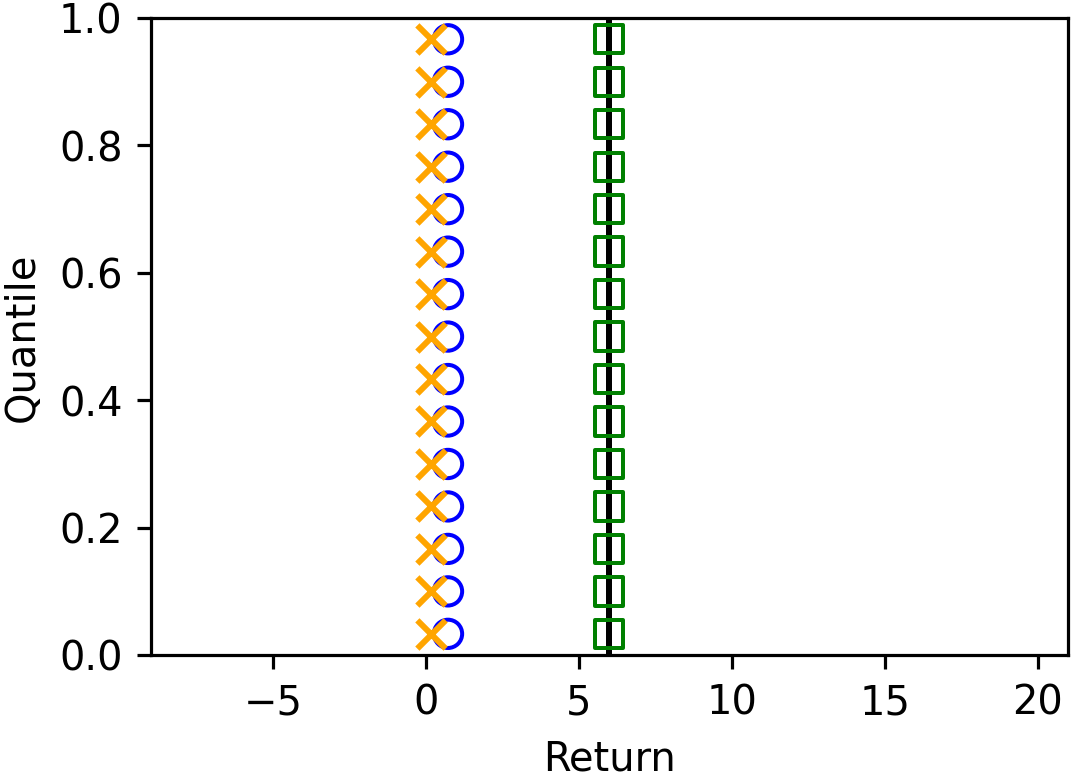}
\label{fig:matrix-game-results-qplex-2}
\vspace{-1.5em}
\subcaption{QPLEX}
\end{minipage}
\begin{minipage}{0.3\textwidth} 
\includegraphics[width=\linewidth]{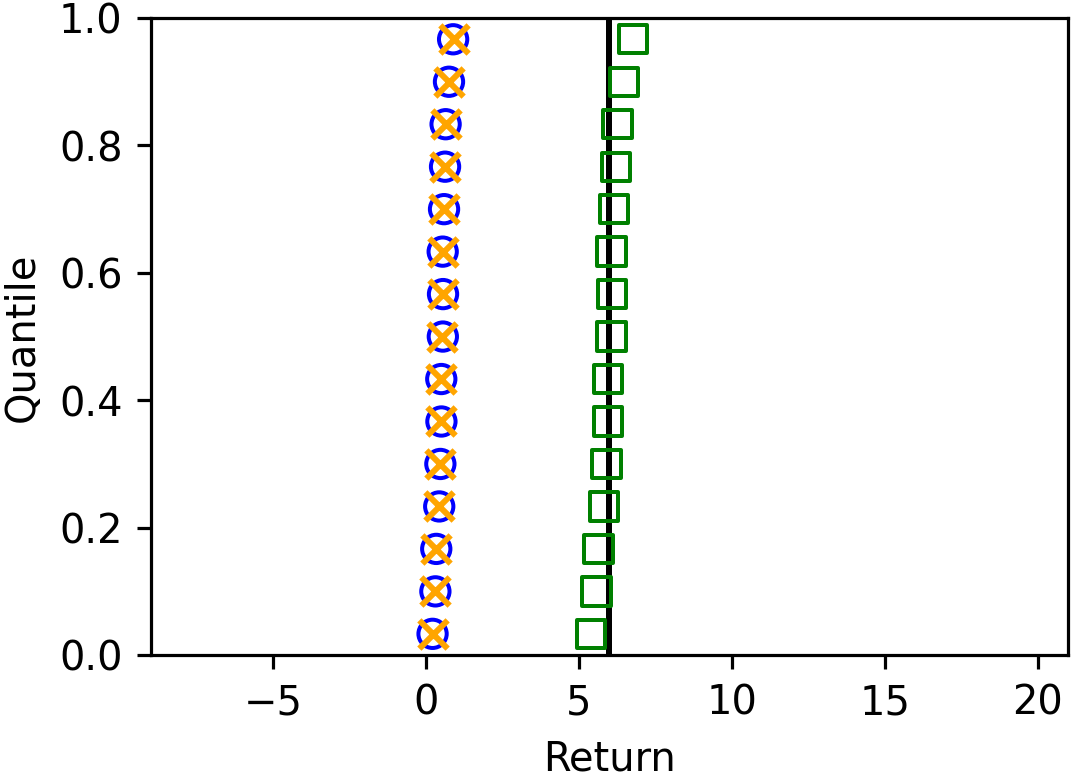}
\label{fig:matrix-game-results-dplex-2}
\vspace{-1.5em}
\subcaption{DPLEX}
\end{minipage}
\begin{minipage}{0.3\textwidth} 
\includegraphics[width=\linewidth]{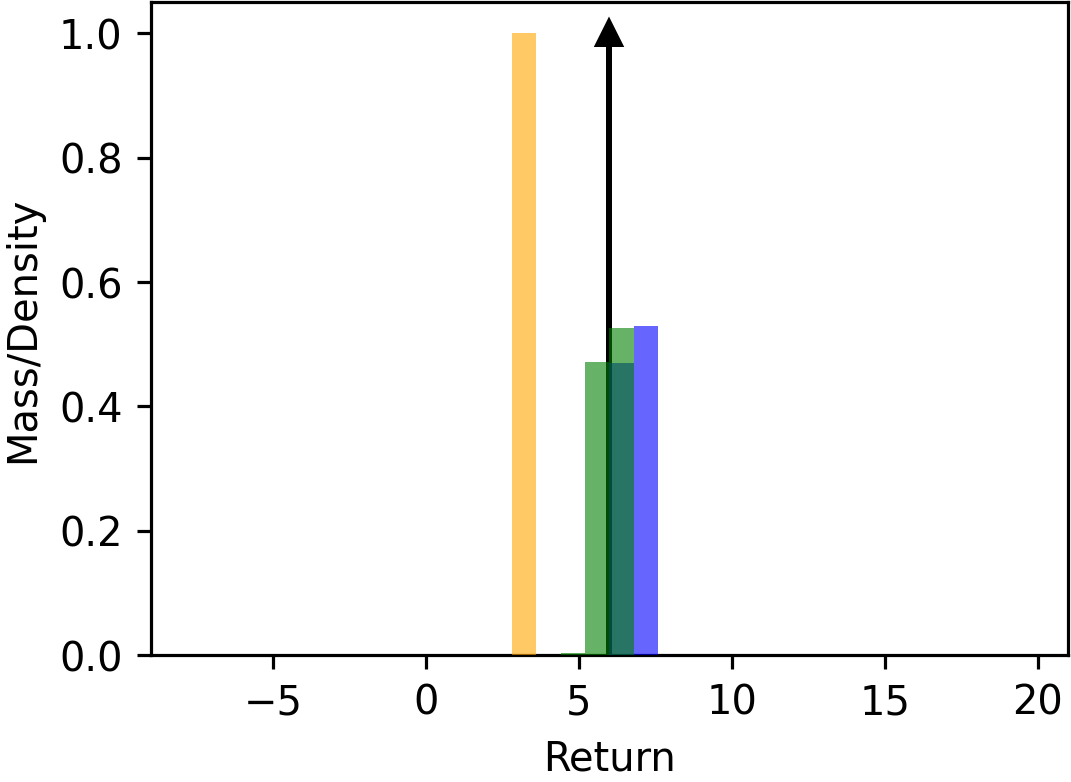}
\label{fig:matrix-game-results-dplex-c51-2}
\vspace{-1.5em}
\subcaption{DPLEX-C51}
\end{minipage}
\end{tabular}
\caption{The learned factorization of the joint values when the joint action $\langle C_1, C_2\rangle$ is selected. The ground truth ($Z_\mathrm{GT}$) of the joint value is the deterministic return $\mathcal{N}(6,0)$).}
\label{fig:matrix-game-results-2}
\end{figure*}

To further illustrate DFAC's capability of factorization, we visualize the learned factorization of the joint values for QPLEX, DPLEX, and DPLEX-C51, respectively, in Figs.~\ref{fig:matrix-game-results-1} and \ref{fig:matrix-game-results-2}. These results demonstrate that QPLEX can only model the shapes of deterministic distributions, while its DFAC variants can model the shapes of both stochastic and deterministic distributions.

Since DPLEX-C51 does not perform satisfactorily in the matrix game (as indicated in Table~\ref{table:matrix-game-results}), we only validate the IQN variants (DDN, DMIX, DPLEX) on the SMAC benchmark in the next section.

\section{Experiment Results on SMAC}
\label{sec:experiment_results}

In this section, we present the experimental results and discuss their implications. We start with a brief introduction to our experimental setup in Section~\ref{subsec:experiment_results_setup_of_smac}. Then, we demonstrate that modeling a full distribution is beneficial to the performance of independent learners in Section~\ref{subsec:experiment_results_independent_learners}. Finally, we compare the performances of the CTDE baseline methods and their DFAC variants in Section~\ref{subsec:experiment_results_super_hard} and Section~\ref{subsec:experiment_results_ultra_hard}.

\begin{figure*}[t]
\begin{tabular}{cc}
\begin{minipage}{1.0\textwidth} 
\includegraphics[width=\linewidth]{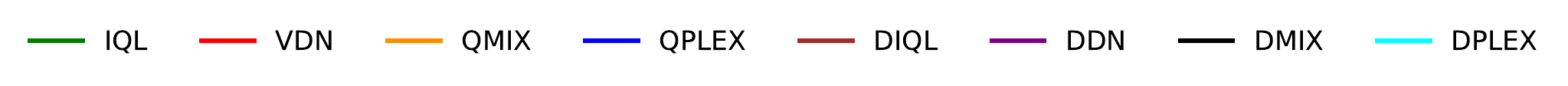}
\vspace{-1.5em}
\label{fig:win_rate_legend}
\end{minipage} \\
\begin{minipage}{0.3\textwidth} 
\includegraphics[width=\linewidth]{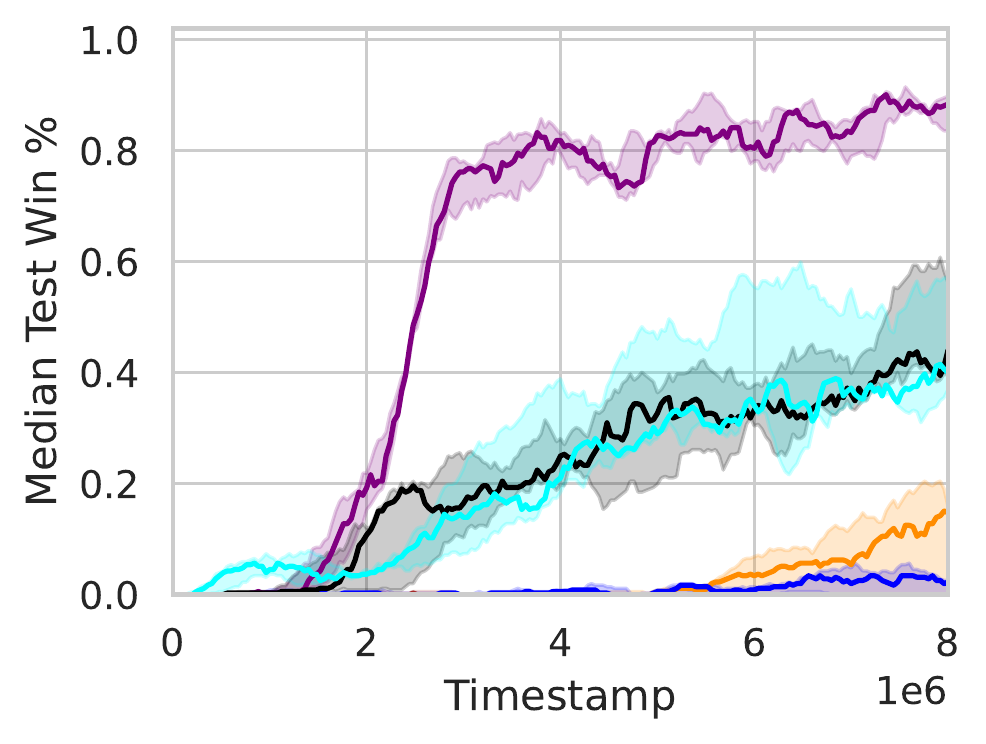}
\label{fig:win_rate_6h_vs_8z}
\vspace{-1.5em}
\subcaption{\texttt{6h\_vs\_8z}}
\end{minipage}
\begin{minipage}{0.3\textwidth} 
\includegraphics[width=\linewidth]{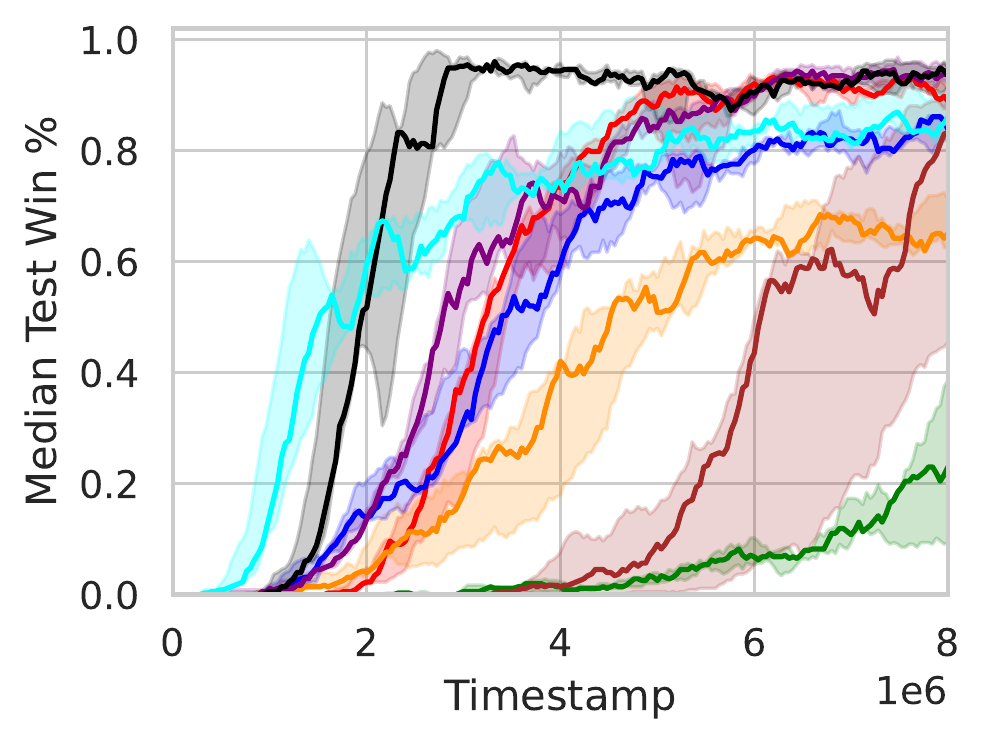}
\label{fig:win_rate_3s5z_vs_3s6z}
\vspace{-1.5em}
\subcaption{\texttt{3s5z\_vs\_3s6z}}
\end{minipage}
\begin{minipage}{0.3\textwidth} 
\includegraphics[width=\linewidth]{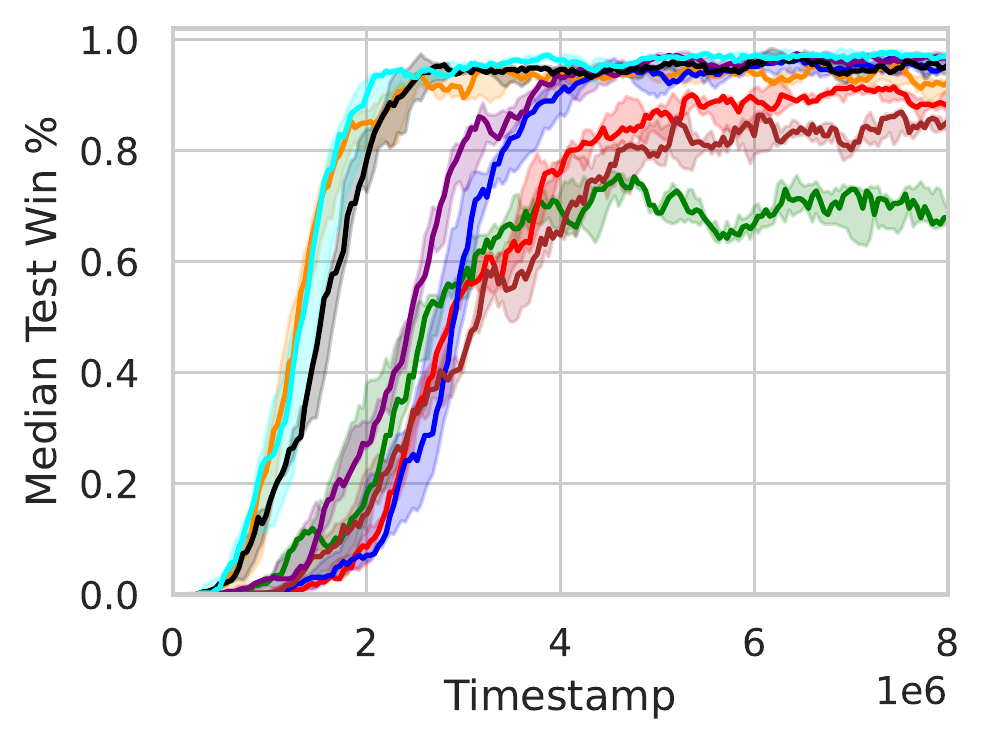}
\label{fig:win_rate_MMM2}
\vspace{-1.5em}
\subcaption{\texttt{MMM2}}
\end{minipage} \\
\begin{minipage}{0.3\textwidth} 
\includegraphics[width=\linewidth]{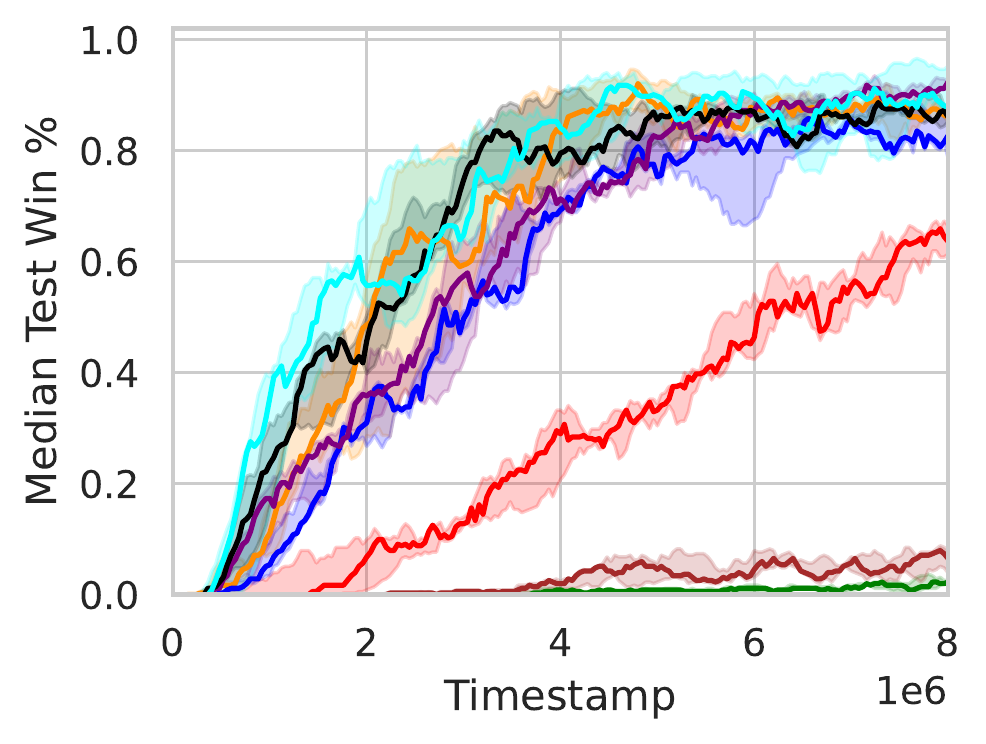}
\label{fig:win_rate_27m_vs_30m}
\vspace{-1.5em}
\subcaption{\texttt{27m\_vs\_30m}}
\end{minipage}
\begin{minipage}{0.3\textwidth} 
\includegraphics[width=\linewidth]{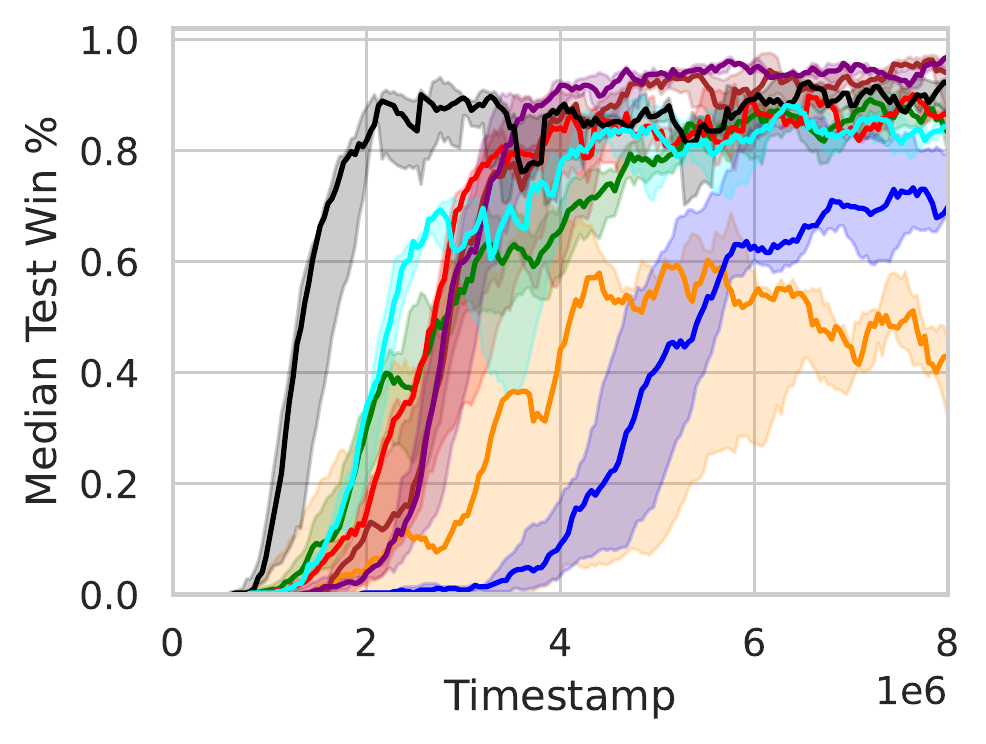}
\label{fig:win_rate_corridor}
\vspace{-1.5em}
\subcaption{\texttt{corridor}}
\end{minipage}
\end{tabular}
\caption{The win rate curves evaluated on the five \superhard{} maps in SMAC.}
\label{fig:smac_results_win_rate}
\end{figure*}

\begin{figure*}[t]
\centering
\begin{tabular}{cc}
\begin{minipage}{0.8\textwidth} 
\includegraphics[width=\linewidth]{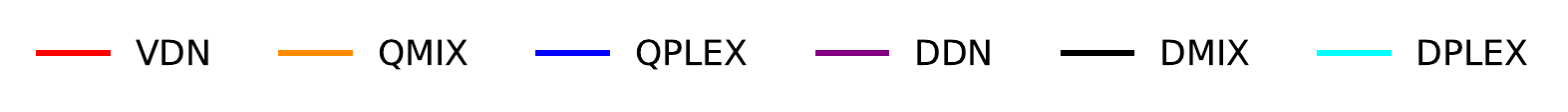}
\vspace{-2em}
\label{fig:win_rate_legend_ultra_hard}
\end{minipage} \\
\begin{minipage}{0.3\textwidth} 
\includegraphics[width=\linewidth]{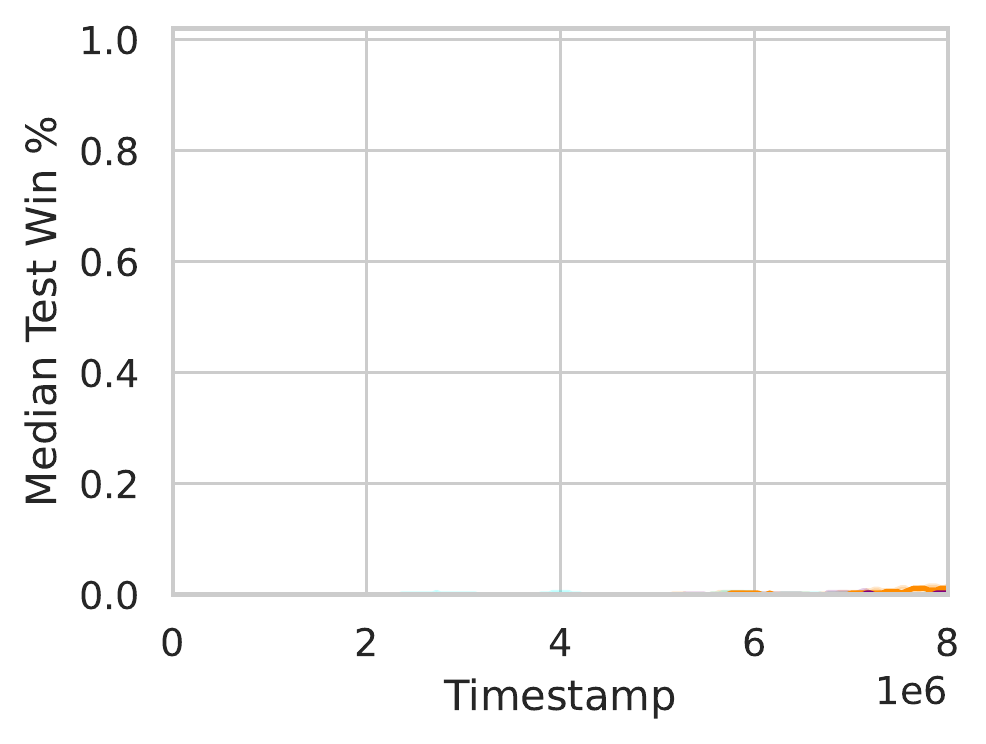}
\label{fig:win_rate_6h_vs_9z}
\vspace{-2em}
\subcaption{\texttt{6h\_vs\_9z}}
\end{minipage}
\begin{minipage}{0.3\textwidth} 
\includegraphics[width=\linewidth]{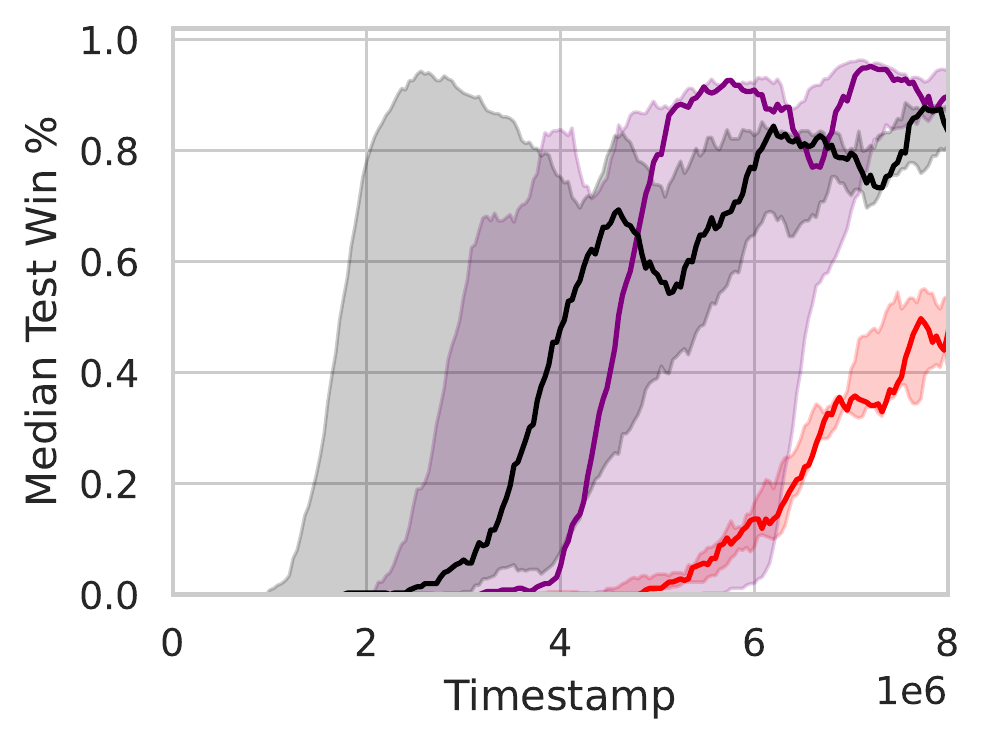}
\label{fig:win_rate_3s5z_vs_4s6z}
\vspace{-2em}
\subcaption{\texttt{3s5z\_vs\_4s6z}}
\end{minipage}
\begin{minipage}{0.3\textwidth} 
\includegraphics[width=\linewidth]{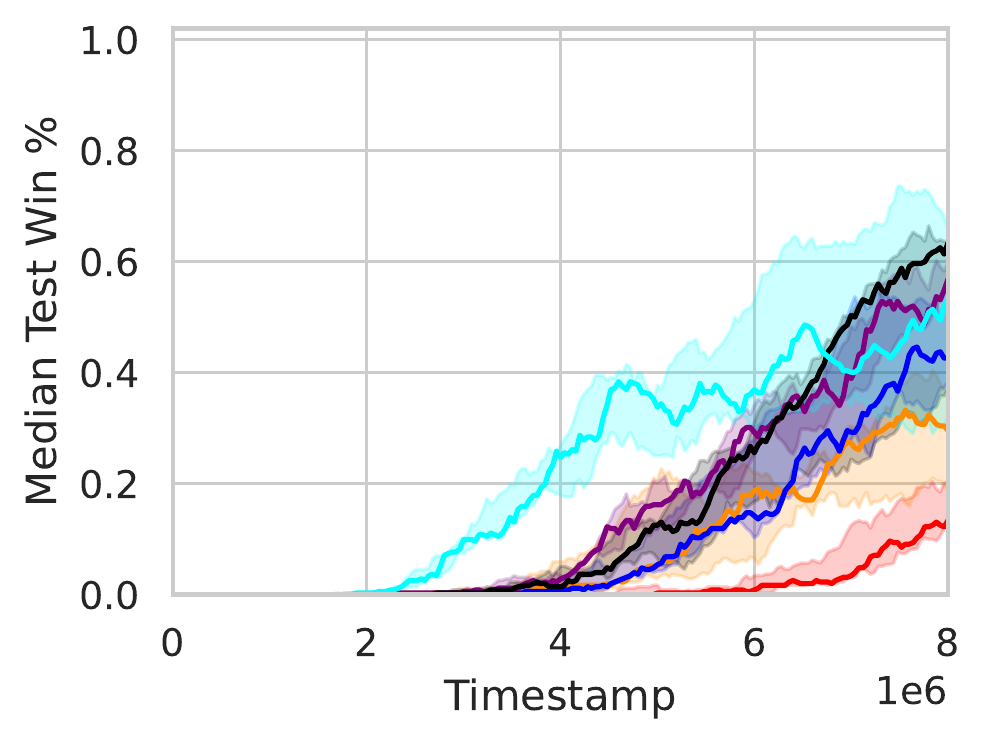}
\label{fig:win_rate_MMM2_7m2M1M_vs_8m4M1M}
\vspace{-2em}
\subcaption{\texttt{MMM2\_7m2M1M\_vs\_8m4M1M}}
\end{minipage} \\
\begin{minipage}{0.3\textwidth}
\includegraphics[width=\linewidth]{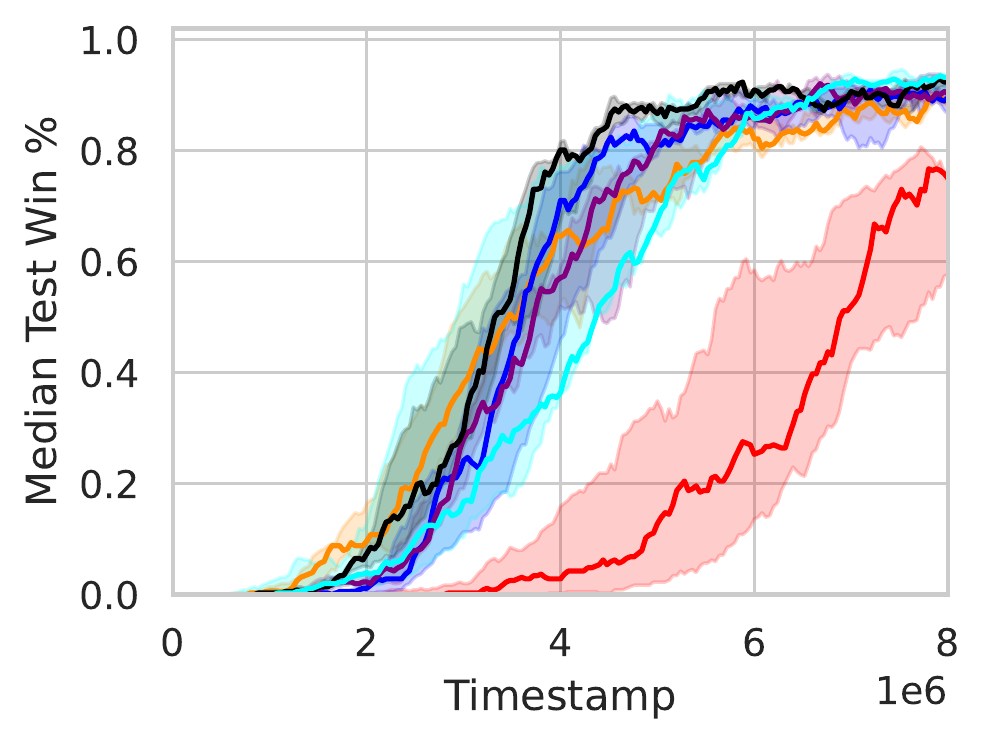}
\label{fig:win_rate_MMM2_7m2M1M_vs_9m3M1M}
\vspace{-2em}
\subcaption{\texttt{MMM2\_7m2M1M\_vs\_9m3M1M}}
\end{minipage}
\begin{minipage}{0.3\textwidth} 
\includegraphics[width=\linewidth]{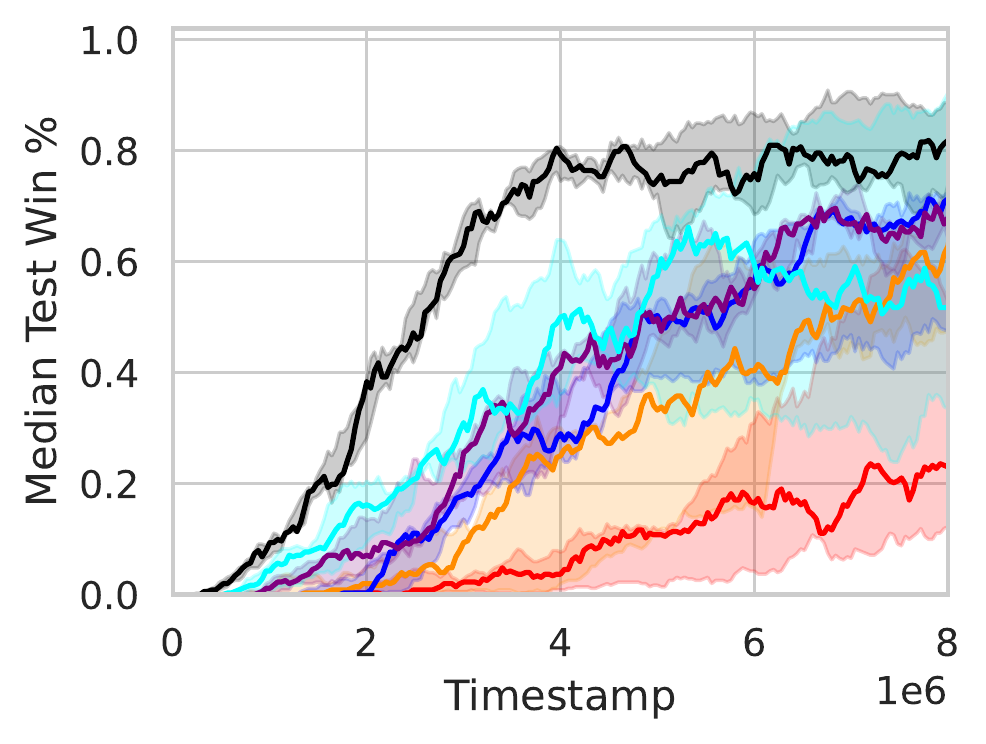}
\label{fig:win_rate_26m_vs_30m}
\vspace{-2em}
\subcaption{\texttt{26m\_vs\_30m}}
\end{minipage}
\begin{minipage}{0.3\textwidth} 
\includegraphics[width=\linewidth]{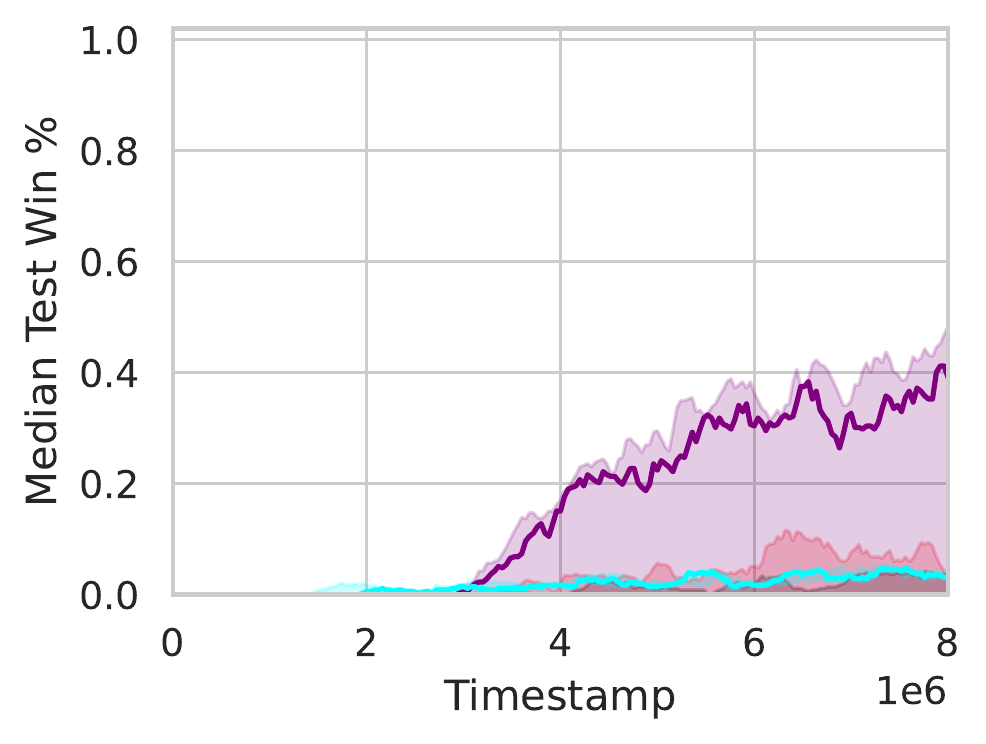}
\label{fig:win_rate_corridor_2z_vs_24zg}
\vspace{-2em}
\subcaption{\texttt{corridor\_2z\_vs\_24zg}}
\end{minipage}
\end{tabular}
\caption{The win rate curves evaluated on the six \ultrahard{} maps.}\vspace{-0.5em}
\label{fig:smac_results_win_rate_ultra_hard}
\end{figure*}


\subsection{Experimental Setup}
\label{subsec:experiment_results_setup_of_smac}

In this section, we describe the environmental setups, the hyperparameter tuning process, as well as the baselines in details.

\textbf{Environment.}
We verify the DFAC framework in the SMAC benchmark environments \citep{Samvelyan2019SMAC} built on the popular real-time strategy game StarCraft II. Instead of playing the full game, SMAC is developed for evaluating the effectiveness of MARL micro-management algorithms. Each environment in SMAC contains two teams. One team is controlled by a decentralized MARL algorithm, with the policies of the agents conditioned on their local observation histories. The other team consists of enemy units controlled by the built-in game artificial intelligence based on carefully handcrafted heuristics, which is set to its highest difficulty equal to seven. The overall objective is to maximize the win rate for each battle scenario, where the rewards employed in our experiments follow the default settings of SMAC. The default settings use \textit{shaped rewards} based on the damage dealt, enemy units killed, and whether the RL agents win the battle. If there is no healing unit in the enemy team, the maximum return of an episode (i.e., the score) is $20$; otherwise, it may exceed $20$, since enemies may receive more damages after healing or being healed.

The environments in SMAC are categorized into three different levels of difficulties: \textit{Easy}, \textit{Hard}, and \textit{Super Hard} scenarios~\citep{Samvelyan2019SMAC}. In this paper, we focus on all \textit{Super Hard} scenarios including (a) \texttt{6h\_vs\_8z},  (b) \texttt{3s5z\_vs\_3s6z}, (c) \texttt{MMM2}, (d) \texttt{27m\_vs\_30m}, and (e) \texttt{corridor}, since these scenarios have not been properly addressed in the previous literature without the use of additional assumptions such as intrinsic reward signals~\citep{Du2019LIIR}, explicit communication channels~\citep{Zhang2019VBN,Wang2019NDQ}, common knowledge shared among the agents~\citep{De2019MACKRL,Wang2020ROMA}, and so on. Four of these scenarios have their maximum scores higher than $20$. In \texttt{3s5z\_vs\_3s6z}, the enemy \textit{Stalkers} have the ability to regenerate shields; in \texttt{MMM2}, the enemy \textit{Medivacs} can heal other units; in \texttt{6h\_vs\_8z} and \texttt{corridor}, the enemy \textit{Zerglings} slowly regenerate their own health. The StarCraft version we used is 4.10.

\textbf{Hyperparameters.}
For all of our experimental results, the training length is set to 8M timesteps, where the agents are evaluated every 40k timesteps with 32 independent runs. The curves presented in this section are generated based on five different random seeds. The solid lines represent the median win rate, while the shaded areas correspond to the $25^{\text{th}}$ to $75^{\text{th}}$ percentiles. For a better visualization, the presented curves are smoothed by a moving average filter with its window size set to 11. 

We tuned the hyperparameters of both the baselines and their distributional variants by selecting their hidden layer sizes from $\{32,64,128,256,512\}$ and choose the best ones. The quantile samples of \diql{} (i.e., a distributional variant of IQL) and \ddn{} are simply set to $\numberofquantiles{}=\numberofquantilesamples{}=1$, since they do not require the calculation of the expected value during the optimization process. As for \dmix{} and DPLEX, the numbers of quantile samples are set to $\numberofquantiles{}=\numberofquantilesamples{}=8$ as in~\cite{Dabney2018IQN}. The optimizers follow those used in DQN and IQN. All of the other hyperparameters follow those used in SMAC. Table~\ref{tab:hyperparameter} lists the hyperparameters adopted for the baselines and their distributional variants.

\textbf{Baselines.}
We select IQL, VDN, QMIX, and QPLEX as our baseline methods, and compare them with their distributional variants in our experiments. The configurations are optimized so as to provide the best performance for each of the methods considered.
Since we tuned the hyperparameters of the baselines, their performances are better than those reported in~\citep{Samvelyan2019SMAC}. The hyperparameter searching process is detailed in the supplementary material.


\subsection{Independent Learners}
\label{subsec:experiment_results_independent_learners}

In order to validate our assumption that distributional RL is beneficial to the MARL domain, we first employ the simplest training algorithm, IQL, and extend it to its distributional variant, called \diql{}. \diql{} is simply a modified IQL that uses IQN as its underlying RL algorithm without any additional modification or enhancements~\citep{Matignon2007Hysteretic,Lyu2020LikelihoodQuantile}.

From Figs.~\ref{fig:smac_results_win_rate}(a)-\ref{fig:smac_results_win_rate}(e) and Tables~\ref{table:smac_results_win_rate_all} and \ref{table:smac_results_score_all}, it is observed that \diql{} is superior to IQL even without utilizing any value function factorization methods. This validates that distributional RL has beneficial influences on MARL, when it is compared to RL approaches based only on expected values.


\subsection{Value Function Factorization Methods}
\label{subsec:experiment_results_super_hard}

In order to inspect the effectiveness and impacts of DFAC on learning curves, win rates, and scores, we next summarize the results of the baselines as well as their DFAC variants on the \textit{Super Hard} scenarios in Fig.~\ref{fig:smac_results_win_rate}(a)-(e) and Tables~\ref{table:smac_results_win_rate_all} and \ref{table:smac_results_score_all}.

Fig.~\ref{fig:smac_results_win_rate}(a)-(e) plot the learning curves of the baselines and their DFAC variants, with the final win rates presented in Table~\ref{table:smac_results_win_rate_all}, and their final scores reported in Table~\ref{table:smac_results_score_all}. The win rates indicate how often do the player's team wins, while the scores represent how well do the player's team performs. Despite the fact that SMAC's objective is to maximize the win rate, the true optimization goal of MARL algorithms is the averaged score. In fact, these two metrics are not always positively correlated (e.g., VDN and QMIX in \texttt{6h\_vs\_8z} and \texttt{3s5z\_vs\_3s6z}, and QMIX and \dmix{} in \texttt{3s5z\_vs\_3s6z}).

It can be observed that the learning curves of \ddn{}, \dmix{}, and \dplex{} grow faster and achieve higher final win rates than their corresponding baselines. In the most difficult map: \texttt{6h\_vs\_8z}, most of the methods fail to learn an effective policy except for the DFAC variants. The evaluation results also show that the DFAC variants are capable of performing consistently well across all \superhard{} maps with high win rates. In addition to the win rates, Table~\ref{table:smac_results_score_all} further presents the final averaged scores achieved by each method, and provides deeper insights into the advantages of the DFAC framework by quantifying the performances of the learned policies of different methods.

The improvements in win rates and scores are due to the benefits offered by distributional RL~\citep{Lyle2019Comparative}, which enables the distributional variants to work more effectively in MARL environments. Moreover, the evaluation results reveal that \ddn{} performs especially well in most environments despite its simplicity.


\subsection{Evaluation Results on the Ultra Hard Maps of SMAC}
\label{subsec:experiment_results_ultra_hard}

Based on the five \superhard{} maps, we further designed six \ultrahard{} maps for evaluating the capability of DFAC. The \ultrahard{} maps are guaranteed to be harder than the \superhard{} maps, since they contain either increased number of units in the enemy team or decreased number of controllable units in the player's team. The detailed comparison between the two sets of maps is presented in Table~\ref{table:smac_maps}.

The evaluation results on the \ultrahard{} maps are reported in Fig.~\ref{fig:smac_results_win_rate_ultra_hard} and Tables~\ref{table:smac_results_win_rate_all} and \ref{table:smac_results_score_all} with five independent runs. It is observed that \ddn{} achieves outstanding performance in most of the maps despite the simplicity of its factorization function. For reproducibility, all \ultrahard{} maps can be found in our GitHub repository (\href{https://github.com/j3soon/dfac-extended}{https://github.com/j3soon/dfac-extended}), along with the gameplay recording videos.

\begin{table*}[t]
\small
\centering
\caption{The median win rate percentage ($\%$) of five independent test runs.}
\begin{tabular}{l|l|rrrr|rrrr}
\toprule
 & Map                     & IQL   & VDN   & QMIX  & QPLEX  & DIQL  & DDN            & DMIX & DPLEX \\
\midrule
\multirow{5}{*}{\rotatebox[origin=c]{90}{\superhard{}}}
& \texttt{6h\_vs\_8z}     &  0.00 &  0.00 & 12.78 &  0.00 &  0.00 & \textbf{83.92} & 49.43 & 43.75 \\
& \texttt{3s5z\_vs\_3s6z} & 29.83 & 89.20 & 67.22 & 84.38 & 62.22 & \textbf{94.03} & 91.08 & 90.62 \\
& \texttt{MMM2}           & 68.92 & 89.20 & 92.44 & 96.88 & 85.23 & \textbf{97.22} & 95.11 & 96.88 \\
& \texttt{27m\_vs\_30m}   &  2.27 & 63.12 & 84.77 & 78.12 &  6.02 & \textbf{91.48} & 85.45 & 90.62 \\
& \texttt{corridor}       & 84.87 & 85.34 & 37.61 & 75.00 & 91.62 & \textbf{95.40} & 90.45 & 81.25 \\
\midrule
\multirow{6}{*}{\rotatebox[origin=c]{90}{\ultrahard{}}}
& \texttt{6h\_vs\_9z}               & - &  0.00 &  \textbf{1.14} & 0.00 & - & 0.28  &  0.00 & 0.00 \\
& \texttt{3s5z\_vs\_4s6z}           & - & 47.16 &  0.00 & 0.00 & - & \textbf{89.77} & 83.52 & 0.00 \\
& \texttt{MMM2\_7m2M1M\_vs\_8m4M1M} & - & 13.35 & 29.55 & 46.88 & - & 56.82 & \textbf{63.35} & 50.00 \\
& \texttt{MMM2\_7m2M1M\_vs\_9m3M1M} & - & 75.00 & 88.64 & 90.62 & - & 90.34 & \textbf{92.33} & 90.62 \\
& \texttt{26m\_vs\_30m}             & - & 23.01 & 62.78 & 78.12 & - & 67.90 & \textbf{81.82} & 59.38 \\
& \texttt{corridor\_2z\_vs\_24zg}   & - &  0.00 &  0.00 & 0.00 & - & \textbf{41.19} &  0.00 & 3.12 \\
\bottomrule
\end{tabular}
\label{table:smac_results_win_rate_all}
\end{table*}

\begin{table*}[t]
\small
\centering
\caption{The averaged scores of five independent test runs.}
\begin{tabular}{l|l|rrrr|rrrr}
\toprule
 & Map                     & IQL   & VDN   & QMIX  & QPLEX  & DIQL  & DDN            & DMIX & DPLEX \\
\midrule
\multirow{5}{*}{\rotatebox[origin=c]{90}{\superhard{}}}
& \texttt{6h\_vs\_8z}     & 13.78 & 15.41 & 14.37 & 15.95 & 14.94 & \textbf{19.40} & 17.14 & 17.88 \\
& \texttt{3s5z\_vs\_3s6z} & 16.54 & 19.75 & 20.16 & 20.42 & 17.52 & \textbf{20.94} & 19.70 & 20.27 \\
& \texttt{MMM2}           & 17.50 & 19.36 & 19.42 & 19.60 & 19.21 & \textbf{20.90} & 19.87 & 19.93 \\
& \texttt{27m\_vs\_30m}   & 14.01 & 18.45 & 19.41 & 19.33 & 14.45 & \textbf{19.71} & 19.43 & 19.62 \\
& \texttt{corridor}       & 19.42 & 19.47 & 15.07 & 18.73 & 19.68 & \textbf{20.00} & 19.66 & 19.08 \\
\midrule
\multirow{6}{*}{\rotatebox[origin=c]{90}{\ultrahard{}}}
& \texttt{6h\_vs\_9z}               & - & 13.57 & 12.37 & 13.86 & - & \textbf{16.00} & 13.73 & 14.84 \\
& \texttt{3s5z\_vs\_4s6z}           & - & 17.16 & 13.09 & 13.60 & - & \textbf{19.65} & 18.61 & 14.99 \\
& \texttt{MMM2\_7m2M1M\_vs\_8m4M1M} & - & 13.13 & 14.40 & 15.52 & - & \textbf{16.50} & 16.24 & 15.89 \\
& \texttt{MMM2\_7m2M1M\_vs\_9m3M1M} & - & 17.30 & 19.01 & 19.06 & - & \textbf{19.45} & 19.33 & 19.40 \\
& \texttt{26m\_vs\_30m}             & - & 16.69 & 18.23 & 18.66 & - & 18.49 & \textbf{19.17} & 18.49 \\
& \texttt{corridor\_2z\_vs\_24zg}   & - &  7.78 &  4.80 & 6.44 & - & \textbf{11.10} &  7.41 & 10.71 \\
\bottomrule
\end{tabular}
\label{table:smac_results_score_all}
\end{table*}

\begin{table*}[t]
\footnotesize
\centering
\caption{A comparison between the \superhard{} maps and the \ultrahard{} maps.}
\begin{tabular}{l|l|l|l}
\toprule
 & Map                     & Player's Team & Enemy's Team \\
\midrule
\multirow{5}{*}{\rotatebox[origin=c]{90}{\superhard{}}}
& \texttt{6h\_vs\_8z}     & 6 Hydralisks & 8 Zealots \\
& \texttt{3s5z\_vs\_3s6z} & 3 Stalkers \& 5 Zealots & 3 Stalkers \& 6 Zealots \\
& \texttt{MMM2}           & 7 Marines, 2 Marauders \& 1 Medivac & 8 Marines, 3 Marauders \& 1 Medivac \\
& \texttt{27m\_vs\_30m}   & 27 Marines & 30 Marines \\
& \texttt{corridor}       & 6 Zealots & 24 Zerglings \\
\midrule
\multirow{6}{*}{\rotatebox[origin=c]{90}{\ultrahard{}}}
& \texttt{6h\_vs\_9z}     & 6 Hydralisks & 9 Zealots \\
& \texttt{3s5z\_vs\_4s6z} & 3 Stalkers \& 5 Zealots & 4 Stalkers \& 6 Zealots \\
& \texttt{MMM2\_7m2M1M\_vs\_8m4M1M} & 7 Marines, 2 Marauders \& 1 Medivac & 8 Marines, 4 Marauders \& 1 Medivac \\
& \texttt{MMM2\_7m2M1M\_vs\_9m3M1M} & 7 Marines, 2 Marauders \& 1 Medivac & 9 Marines, 3 Marauders \& 1 Medivac \\
& \texttt{26m\_vs\_30m}   & 26 Marines & 30 Marines \\
& \texttt{corridor\_2z\_vs\_24zg} & 2 Zealots & 24 Zerglings \\
\bottomrule
\end{tabular}
\label{table:smac_maps}
\end{table*}


\subsection{Computational Infrastructure}

In our experiments, NVIDIA DGX-1 clusters were used as the fundamental computing infrastructure, where each DGX-1 contains eight V100 GPUs, two 20-core CPUs, and 512 GB memory. Our experiments were performed on docker instances, where each instance is allocated with 50 GB DDR4 memory, an eight-core Xeon E5-2698 CPU, 1 TB SSD storage, and either a V100-16G or a V100-32G GPU. On an average, the experiment for each map takes around forty hours to train for eight million timesteps, depending on the type of the SMAC map and the chosen algorithm.

\section{Discussions and Outlook}

In this section, we provide discussions on the SMAC benchmark and propose some future directions to extend DFAC.

\textbf{The Performance Metric.} In the original SMAC paper, the authors proposed to use the median test win rates as the main performance metric. However, their optimization goal is actually the discounted scores instead of the win rates. Therefore, the averaged scores of multiple test runs are better metrics for comparing the performance of different methods.

\textbf{Enemy Sight Range.} By inspecting the learned policies of \ddn{}, we observed that the agents can exploit enemies' limited sight ranges by luring out and eliminating a few enemies while keeping distances from the remainings in \texttt{corridor\_2z\_vs\_24zg}. To make the map more challenging, the enemies can be configured to have infinite sight ranges when designing the map.

\textbf{The Reward Function.} The reward function is defined as the accumulated damages dealt to the enemies. However, the enemies may regenerate their health and shields. This may lead to a scenario that the agents intentionally stop attacking the enemies with low health and attack them later after their health is regenerated, allowing the agents to collect more rewards. Such a behavior can be prevented by providing no reward signal for regenerated health.

\textbf{Expressiveness and Performance.} 
According to the results presented in Section~\ref{sec:experiment_results}, the \superhard{} maps can be solved under the assumption of \additivity{}, since \ddn{} is able to achieve high win rates on all of these maps.
These results indicate that the other methods with additional expressiveness have no theoretical advantage over \ddn{} on these \superhard{} maps, and may perform worse due to the instability caused by their factorization functions. This further verifies the fact that the expressiveness is not necessarily correlated with the performance in SMAC, which is often misinterpreted in the current literature.

\textbf{Network Size and Training Steps.} 
In the original SMAC paper, the authors did not tune the size of the policy networks and only trained them for two million timesteps. According to our experiments, we found that using larger policy networks and training them for more timesteps does improve the performance by a noticeable margin.

\textbf{Future Extensions.} Based on the DFAC framework, many SARL techniques that require the shapes of the return distributions may be incorporated into the MARL domain, such as exploration~\citep{Nikolov2019IDS,Zhang2019QUOTA,Mavrin2019DLTV} and risk-aware policies~\citep{Xia2020risk}. Aside from the simple shape function used in this work, future endeavors may include more complex shape functions with learnable weights. It would also be interesting to incorporate other distributional RL algorithms into the DFAC framework aside from C51 and IQN, such as quantile regression DQN (QR-DQN)~\citep{Dabney2018QR-DQN}, expectile distributional RL (EDRL)~\citep{Rowland2019ER-DQN}, fully parameterized quantile function (FQF)~\citep{Yang2019FQF}, and maximum mean discrepancy RL (MMDRL)~\citep{Nguyen2021MMDRL}.

\begin{table*}[t]
\small
\centering
\caption{A summary of the optimal hidden state sizes of the baseline methods and their distributional variants.}
\begin{tabular}{l|l|cccc|cccc}
\toprule
 & Map & IQL & VDN & QMIX & QPLEX & DIQL & DDN & DMIX & DPLEX \\
\midrule
\multirow{5}{*}{\rotatebox[origin=c]{90}{\superhard{}}}
& \texttt{6h\_vs\_8z}     & 128 & 128 & 256 & 256 & 512 & 512 & 256 & 512 \\
& \texttt{3s5z\_vs\_3s6z} & 512 & 128 & 128 & 128 & 256 & 512 & 256 & 512 \\
& \texttt{MMM2}           & 256 & 64  & 64  & 64  & 512 & 512 & 256 & 256 \\
& \texttt{27m\_vs\_30m}   & 256 & 64  & 64  & 64  & 512 & 128 & 128 & 128 \\
& \texttt{corridor}       & 256 & 128 & 256 & 64  & 512 & 128 & 64  & 512 \\
\midrule
\multirow{6}{*}{\rotatebox[origin=c]{90}{\ultrahard{}}}
& \texttt{6h\_vs\_9z}               & - & 256 & 64  & 256 & - & 128 & 512 & 256 \\
& \texttt{3s5z\_vs\_4s6z}           & - & 64  & 64  & 128 & - & 512 & 512 & 256 \\
& \texttt{MMM2\_7m2M1M\_vs\_8m4M1M} & - & 128 & 128 & 128 & - & 256 & 128 & 256 \\
& \texttt{MMM2\_7m2M1M\_vs\_9m3M1M} & - & 64  & 64  & 128 & - & 256 & 128 & 128 \\
& \texttt{26m\_vs\_30m}             & - & 32  & 32  & 64 & - & 128 & 64  & 128 \\
& \texttt{corridor\_2z\_vs\_24zg}   & - & 128 & 512 & 256 & - & 256 & 512 & 512 \\
\bottomrule
\end{tabular}
\label{tab:hyperparameter}
\end{table*}

\section{Conclusion}
\label{sec:conclusion}

In this paper, we introduced a unified framework for cooperative MARL, called DFAC, for integrating distributional RL with value function factorization. DFAC is based on mean-shape decomposition to meet the Distributional IGM condition.
To realize DFAC in a computationally friendly manner, its shape function is implemented as convolution for C51, and quantile mixture for IQN.
DFAC's ability to factorize a joint return distribution into individual utility distributions was demonstrated in a matrix game. In order to validate the effectiveness of DFAC, we presented experimental results performed on all \superhard{} scenarios in SMAC for a number of MARL baseline methods as well as their DFAC variants. Moreover, we performed experiments on a number of self-designed \ultrahard{} maps to further validate the effectiveness of DFAC. The results showed that in most of the scenarios, \ddn{}, \dmix{}, and \dplex{} outperform VDN, QMIX, and QPLEX, respectively. DFAC can be extended to other value function factorization methods and offers an interesting research direction for future endeavors.

\acks{The authors gratefully acknowledge the support from the National Science and Technology Council (NSTC) in Taiwan under grant number MOST 111-2223-E-007-004-MY3. The authors would also like to express their appreciation for the donation of the GPUs from NVIDIA Corporation and NVIDIA AI Technology Center (NVAITC) used in this work. In addition, the authors extend their gratitude to the National Center for High-Performance Computing (NCHC) for providing the necessary computational and storage resources. Furthermore, the authors thank Kuan-Yu Chang for his helpful critiques of this research work.}

\newpage

\vskip 0.2in
\bibliography{references}

\end{document}